\documentclass[10pt,english]{article}
\usepackage[T1]{fontenc}
\usepackage[latin9]{inputenc}
\usepackage{geometry}
\usepackage{amsthm}
\usepackage{amsmath}
\usepackage{amssymb}
\usepackage{graphicx}
\usepackage{setspace}
\usepackage[authoryear]{natbib}
\makeatletter
\numberwithin{equation}{section}

\numberwithin{equation}{section}
\numberwithin{figure}{section}
\theoremstyle{plain}
\newtheorem{thm}{\protect\theoremname}[section]
  \theoremstyle{remark}
  \newtheorem{rem}{\protect\remarkname}[section]
  \theoremstyle{plain}
  \newtheorem{lem}{\protect\lemmaname}[section]
  \theoremstyle{plain}
  \newtheorem{prop}{\protect\propositionname}[section]

\usepackage{stmaryrd}

\newcounter{hypA}
\newenvironment{condition}{\refstepcounter{hypA}\begin{itemize}
\item[({\bf H\arabic{hypA}})]}{\end{itemize}}

\usepackage{babel}\providecommand{\lemmaname}{Lemma}
  \providecommand{\propositionname}{Proposition}
  \providecommand{\remarkname}{Remark}
\providecommand{\theoremname}{Theorem}

\usepackage{babel}

\makeatother

\usepackage{babel}
\begin{document}

\title{Linear Variance Bounds for Particle Approximations of Time-Homogeneous
Feynman-Kac Formulae}

\author{\textsc{NICK WHITELEY, NIKOLAS KANTAS \& AJAY JASRA}\\
 \emph{\small Department of Mathematics, University of Bristol, Bristol,
BS8 1TW, UK.}{\small }\\
{\small{} {} {} }\emph{\small Department of Electrical \& Electronic
Engineering, Imperial College London, London, SW7 2AZ, UK.}{\small }\\
{\small{} {} {} }\emph{\small Department of Statistics \& Applied Probability, National University of Singapore, Singapore, 117546, Sg.}}
\maketitle
\begin{abstract}
This article establishes sufficient conditions for a linear-in-time
bound on the non-asymptotic variance for particle approximations of
time-homogeneous Feynman-Kac formulae. These formulae appear in a
wide variety of applications including option pricing in finance
and risk sensitive control in engineering. In direct Monte Carlo
approximation of these formulae, the non-asymptotic variance typically
increases at an exponential rate in the time parameter. It is shown
that a linear bound holds when a non-negative kernel, defined by
the logarithmic potential function and Markov kernel which specify
the Feynman-Kac model, satisfies a type of multiplicative drift condition
and other regularity assumptions. Examples illustrate that these conditions
are general and flexible enough to accommodate two rather extreme
cases, which can occur in the context of a non-compact
state space: 1) when the potential function is bounded above, not
bounded below and the Markov kernel is not ergodic; and 2) when the
potential function is not bounded above, but the Markov kernel itself
satisfies a multiplicative drift condition. 
\end{abstract}

\section{Introduction}

On a state space $\mathsf{X}$ endowed with a $\sigma$-algebra $\mathcal{B}\left(\mathsf{X}\right)$
let $M$ be a Markov kernel and let $U:\mathsf{X}\rightarrow\mathbb{R}$
be a logarithmic potential function. Then for $x\in\mathsf{X}$, consider
the sequence of measures $\left\{ \gamma_{n,x};n\geq1\right\} $ defined
by 
\begin{eqnarray}
\gamma_{n,x}\left(\varphi\right) & := & \mathbb{E}_{x}\left[\exp\left(\sum_{k=0}^{n-1}U\left(X_{k}\right)\right)\varphi\left(X_{n}\right)\right],\label{eq:FK_formula}
\end{eqnarray}
for a suitable test function $\varphi$ and where $\mathbb{E}_{x}$
denotes expectation with respect to the law of a Markov chain $\left\{ X_{n};n\geq0\right\} $
with transition kernel $M$, initialised from $X_{0}=x$. 

Feynman-Kac
formulae as in (\ref{eq:FK_formula}) arise in a variety of application
domains. In the case that $U$ is non-positive, the quantity $\gamma_{n,x}\left(1\right)$
can be interpreted as the probability of survival up to time step
$n$ of a Markovian particle exploring an absorbing medium \citep{smc:the:dMM03,smc:the:DMD04};
the particle evolves according to $M$ and at time step $k$ it is
killed with probability $1-\exp\left(U\left(X_{k}\right)\right)$. Another
application is the calculation of expectations at a terminal time
with respect to jump-diffusion processes which may or may not be partially
observed (e.g.~\citet{jasra1}). In particular, for option pricing
in finance, there are a variety of options, (e.g.~asian, barrier)
which can be written in the form \eqref{eq:FK_formula} where the
potential function arises from the pay-off function/change of measure
and the Markov kernel specifies finite dimensional marginals of some
partially observed L\'evy process (e.g.~\citet{jasra}). It is remarked
that in this latter example, the finite dimensional marginals can
induce a time-homogeneous Markov chain that is not necessarily ergodic.
Furthermore, functionals as in (\ref{eq:FK_formula}) arise in certain
stochastic control problems, where one considers the bivariate process
$\{X_{n}=(Y_{n},A_{n});n\geq0\}$ with $Y_{n}$ being a controlled
Markov chain and $\{A_{n};n\geq0\}$ a control input process. In some
cases the transition kernel $M$ can be expressed as $M_{1}(y_{n},da_{n})M_{2}(y_{n},a_{n},dy_{n+1})$
with $M_{1}$ corresponding to the control law or policy and $M_{2}$
to the controlled process dynamics. In a risk-sensitive optimal control
framework $\frac{1}{n}\log\gamma_{n,x}\left(1\right)$ arises as a
cost function one aims to minimise with respect to an appropriate
class of policies; see \citep{Whittle:RiskSens,DiMasiStettner:SICON:1999}
for details. In such problems it is common to choose $U(y,a)$ to
be unbounded from above, e.g. $U$ is usually chosen to be a quadratic
for linear and Gaussian state space models \citep{Whittle:RiskSens}.
More generally (\ref{eq:FK_formula}) arises as a special case of
a time-inhomogeneous Feynman-Kac formulae studied by \citet{smc:theory:Dm04}.

The non-negative kernel $Q\left(x,dy\right):=\exp\left(U(x)\right)M(x,dy)$,
defines a linear operator on functions $Q\left(\varphi\right)(x):=\int Q\left(x,dy\right)\varphi\left(y\right)$
and (\ref{eq:FK_formula}) can be rewritten as $\gamma_{n,x}\left(\varphi\right)=Q_{n}\left(\varphi\right)(x)$,
where $Q_{n}$ denotes the $n$-fold iterate of $Q$. In the applications
described above, the Feynman-Kac formulae (\ref{eq:FK_formula})
typically cannot be evaluated analytically. However, they may be approximated
using a system of interacting particles \citep{smc:theory:Dm04}.
These particle systems, also known as sequential Monte Carlo methods
in the computational statistics literature (e.g.~\citet{smc:meth:DDG2001}),
have themselves become an object of intensive study, see amongst others
\citep{smc:theory:CB08,smc:the:dMPR09,smc:the:vH09,smc:the:CRDM11,smc:the:DMJDJ11}
and references therein for recent developments in a variety of settings.

The present work is concerned with second moment properties of errors associated
with the particle approximations of $\left\{\gamma_{n,x}\right\}$. 
In order to obtain bounds on the relative variance, we control certain
tensor-product functionals of these particle approximations, recently
addressed by \citet{smc:the:CdMG11}, using stability properties of
the operators $\left\{ Q_{n};n\geq1\right\}$. These stability properties
are themselves derived from the multiplicative ergodic and spectral
theories of linear operators on weighted $\infty$-norm spaces due
to \citet{mc:theory:KM03,mc:the:KM05}; this is one of the main novelties
of the paper. By doing so we obtain a linear-in-$n$ relative variance
bound under assumptions on $Q$ which are weaker than those relied upon in
the literature to date and which readily hold on non-compact spaces.
Furthermore, to the knowledge of the authors, these are the first
results which establish 
\begin{itemize}
\item {that a linear-in-$n$ bound holds under conditions which can accommodate
$Q$ defined in terms of a non-ergodic Markov kernel $M$,} 
\item {that any form of non-asymptotic stability result for particle approximations
of Feynman Kac formulae holds under conditions which can accommodate
$U$ not bounded above.} 
\end{itemize}

\subsection{Interacting Particle Systems\label{sub:Particle-systems}}

Let $N\in\mathbb{N}$ be a population size parameter. For $n\in\mathbb{N}$,
let $\zeta_{n}^{(N)}:=\left\{ \mbox{\ensuremath{\zeta}}_{n}^{\left(N,i\right)};1\leq i\leq N\right\} $
be the $n$-th generation of the particle system, where each particle,
$\mbox{\mbox{\ensuremath{\zeta}}}_{n}^{\left(N,i\right)}$, is a random
variable valued in $\mathsf{X}$. Denote $\eta_{n}^{N}:=\dfrac{1}{N}\sum_{i=1}^{N}\delta_{\mbox{\ensuremath{\zeta}}_{n}^{\left(N,i\right)}}$.
The generations of the particle system $\left\{ \zeta_{n}^{(N)};n\geq0\right\} $
form a $\mathsf{X}^{N}$-valued Markov chain: for $x\in\mathsf{X}$,
the law of this chain is denoted by $\mathbb{P}_{x}^{N}$ and has
transitions given in integral form by: 
\begin{eqnarray}
\mathbb{P}_{x}^{N}\left(\zeta_{0}^{(N)}\in dy\right) & = & \prod_{i=1}^{N}\delta_{x}\left(dy^{i}\right),\nonumber \\
\nonumber \\
\mathbb{P}_{x}^{N}\left(\left.\zeta_{n}^{(N)}\in dy\right|\zeta_{n-1}^{(N)}\right) & = & \prod_{i=1}^{N}\left(\frac{\eta_{n-1}^{N}Q(dy^{i})}{\eta_{n-1}^{N}Q(1)}\right),\quad n\geq1,\label{eq:particle_transitions}
\end{eqnarray}
 where $dy=d\left(y^{1},\ldots y^{N}\right)$, $1$ is the unit function
and for some test function $\varphi$, $\eta_{n}^{N}\left(\varphi\right):=\dfrac{1}{N}\sum_{i=1}^{N}\varphi\left(\mbox{\ensuremath{\zeta}}_{n}^{\left(N,i\right)}\right)$
(here the dependence of $\eta_{n}^{N}$ on $x$ is suppressed from
the notation). These transition probabilities correspond to a simple
selection-mutation operation: at each time step $N$ particles are
selected with replacement from the population, on the basis of {}``fitness''
defined in terms of $e^{U}$, followed by each particle mutating
in a conditionally-independent manner according to $M$.

The empirical measures $\left\{ \gamma_{n,x}^{N};n\geq0\right\} $,
defined by 
\[
\gamma_{n,x}^{N}\left(\varphi\right):=\prod_{k=0}^{n-1}\eta_{k}^{N}\left(e^{U}\right)\eta_{n}^{N}\left(\varphi\right),\quad n\geq1,
\]
 and $\gamma_{0,x}^{N}:=\delta_{x},$ are taken as approximations
of $\left\{ \gamma_{n,x}\right\} $. It is well known \citep[Chapter 9]{smc:theory:Dm04} that \[\mathbb{E}_{x}^{N}\left[\gamma_{n,x}^{N}\left(\varphi\right)\right]=\gamma_{n,x}\left(\varphi\right),\]
where $\mathbb{E}_{x}^{N}$ denotes
expectation with respect to the law of the $N$-particle system.

\subsection{Standard Regularity Assumptions for Stability}

\label{sec:standard_reg_assumptions}

Recent work on analysis of tensor product functionals associated with
$\left\{ \gamma_{n,x}^{N};n\geq0\right\} $, \citep{smc:the:dMPR09},
has lead to important results regarding higher moments of the error
associated with these particle approximations; in a possibly time-inhomogeneous
context \citet{smc:the:CdMG11} have proved a remarkable linear-in-$n$
bound on the relative variance of $\gamma_{n,x}^{N}(1)$. In the context
of time-homogeneous Feynman-Kac models, the assumptions of \citet{smc:the:CdMG11}
are that 
\begin{eqnarray}
\sup_{x\in\mathsf{X}}U(x) & < & \infty\label{eq:U_bounded}
\end{eqnarray}
 and that for some $m_{0}\geq1$, there exists a finite constant $c$
such that 
\begin{eqnarray}
Q_{m_{0}}\left(x,dy\right) & \leq & cQ_{m_{0}}\left(x',dy\right),\quad\forall\left(x,x'\right)\in\mathsf{X}^{2}.\label{eq:DV_condition}
\end{eqnarray}
 The result of \citet{smc:the:CdMG11} is then of the form: 
\begin{align}
N> c\left(n+1\right)& &\Longrightarrow& &\mathbb{E}_{x}^{N}\left[\left(\frac{\gamma_{n,x}^{N}(1)}{\gamma_{n,x}(1)}-1\right)^{2}\right] \leq  c\frac{4}{N}\left(n+1\right),\quad\forall x\in\mathsf{X}.\label{eq:Cerou_intro_result}
\end{align}
where $c$ is as in \eqref{eq:DV_condition}. The efficiency of the
particle approximation is therefore quite remarkable: a natural alternative
scheme for estimation of $\gamma_{n,x}(1)$ is to simulate $N$ independent
copies of the Markov chain with transition $M$ and approximate the
expectation in \eqref{eq:FK_formula} by simple averaging, but the
relative variance in that case typically explodes \emph{exponentially}
in $n$. The restriction is that (\ref{eq:DV_condition}) rarely holds
on non-compact spaces. The present work is concerned with proving
a result of the same form as \eqref{eq:Cerou_intro_result} under
assumptions which are more readily verifiable when $\mathsf{X}$ is
non-compact. The main result is summarized after the following discussion
of (\ref{eq:U_bounded})-(\ref{eq:DV_condition}) and how they relate
to the assumptions we consider.

The condition of (\ref{eq:DV_condition}) and its variants are very
common in the literature on exponential stability of nonlinear filters and their particle approximations, see for example \citep{smc:the:DMG01,smc:the:LGO04}
and references therein. It can be interpreted as implying a uniform
bound on the relative oscillations of the total mass of $Q_{m_{0}}$,
i.e., 
\begin{equation}
\frac{Q_{m_{0}}\left(1\right)\left(x\right)}{Q_{m_{0}}\left(1\right)\left(x'\right)}\leq c,\quad\forall\left(x,x'\right)\in\mathsf{X}^{2},\label{eq:Q_osc_bound}
\end{equation}
 and this is very useful when controlling various functionals which
arise when analysing the relative variance as in (\ref{eq:Cerou_intro_result}),
(see \citealt[Proof of Theorem 5.1]{smc:the:CdMG11}). However one
may take the interpretation of (\ref{eq:DV_condition}) in another
direction: it implies immediately that there exist finite measures,
say $\beta$ and $\nu$, and $\epsilon>0$ such that 
\begin{equation}
Q_{m_{0}}\left(x,dy\right)\leq\beta\left(dy\right),\quad Q_{m_{0}}\left(x,dy\right)\geq\epsilon\nu\left(dy\right),\quad\forall x\in\mathsf{X}.\label{eq:minor_strong}
\end{equation}
 In the case that $U=0$ (i.e $Q=M$ is a probabilistic kernel) and
$M$ is $\psi$-irreducible and aperiodic, this type of minorization
over the entire state space $\mathsf{X}$ implies uniform ergodicity
of $Q$, which is in turn equivalent to $Q$ satisfying a Foster-Lyapunov
drift condition with a bounded drift function \citep[Theorem 16.2.2]{mc:theory:MT09}.
In the scenario of present interest, where in general $U\neq0$, one
may take $V:\mathsf{X}\rightarrow[1,\infty)$ to be defined by $V(x)=1$,
for all $x$, and then when (\ref{eq:U_bounded}) holds, it is trivially
true that there exists $\delta\in\left(0,1\right)$ and $b<\infty$
such that $Q$ satisfies the \emph{multiplicative} drift condition,
\begin{eqnarray}
Q\left(e^{V}\right) & \leq & e^{V(1-\delta)+b\mathbb{I}_{\mathsf{X}}},\label{eq:drift_bounded}
\end{eqnarray}
 where $\mathbb{I}_{\mathsf{X}}$ is the indicator function on $\mathsf{X}$.
$Q$ may then also be viewed as a bounded linear operator on the space
of real-valued and bounded functions on $\mathsf{X}$ endowed with
the $\infty$-norm, which is norm-equivalent \citep[in the sense of][p.393]{mc:theory:MT09}
to $\left\Vert \varphi\right\Vert _{e^{V}}:=\sup_{x\in\mathsf{X}}\dfrac{\left|\varphi(x)\right|}{\exp V(x)}$,
with $V$ any bounded weighting function.

As explained in the next section, the interest in writing (\ref{eq:minor_strong})-(\ref{eq:drift_bounded})
is that conditions expressed in this manner have natural
generalisations in the context of weighted $\infty$-norm function
spaces with possibly unbounded $V$.

\subsection{Setting and Main Result}

\citet[(e.g. Chapter 4 and Section 12.4)]{smc:theory:Dm04} and \citet{smc:the:DMD04}
address the setting in which $\left\{ Q_{n};n\geq1\right\} $ is considered
as a semigroup of bounded linear operators on the Banach space of
real-valued and bounded functions on $\mathsf{X}$, endowed with the
$\infty$-norm, and \citet{smc:the:dMM03} address the $L_{2}$ setting,
connecting stability properties of the measures $\left\{ \gamma_{n,x}\right\} $
and their normalized counterparts to the spectral theory of bounded
linear operators on Banach spaces.

\citet{mc:theory:KM03,mc:the:KM05} have developed multiplicative
ergodic and spectral theories of operators of the form $Q$ in the
setting of \emph{weighted} $\infty$-norm spaces; a function space
setting which has already proved to be very fruitful for the study
of general state-space Markov chains \citep[Chapter 16]{mc:theory:MT09}
without reversibility assumptions. The reader is referred to \citep{mc:theory:KM03,mc:the:KM05}
for extensive historical perspective on this spectral theory and related
topics, including (of particular relevance in the present context)
the theory of non-negative operators due to \citet[Chapter 5]{mc:the:N04}.
The work of \citep{mc:theory:KM03,mc:the:KM05,mc:the:M06} is geared
towards large deviation theory for sample path ergodic averages $n^{-1}\sum_{k=0}^{n-1}U(X_{k})$
under the transition $M$ and in that context it is natural to state
assumptions on $M$ and $U$ separately. By contrast, when studying
the particle systems described above, we are not directly concerned
with such sample paths, but rather the relationship between the properties
of the particle approximations $\left\{ \gamma_{n,x}^{N}\right\} $
and their exact counterparts $\left\{ \gamma_{n,x}\right\} $. Some
of the results of \citet{mc:theory:KM03,mc:the:KM05} will be applied
to this effect, but starting from assumptions expressed directly in
terms of $Q$ which reflect the scenario of interest.

The core assumptions in the present work (see Section \ref{sub:Multiplicative-ergodic-theorem}
for precise statements) are that for some constants $m_{0}\geq1$,
$\delta\in(0,1)$ and all $d\geq1$ large enough, 
\begin{eqnarray}
Q_{m_{0}}\left(x,dy\right) & \geq & \epsilon_{d}\nu_{d}\left(dy\right),\quad\forall x\in C_{d},\label{eq:intro_minor}\\
Q\left(e^{V}\right) & \leq & e^{V(1-\delta)+b_{d}\mathbb{I}_{C_{d}}},\label{eq:intro_drift}
\end{eqnarray}
 with $V$ unbounded and $C_{d}:=\left\{ x:V(x)\leq d\right\} \subset\mathsf{X}$
a sublevel set. It is noted that one recovers the minorization and
drift of (\ref{eq:minor_strong})-(\ref{eq:drift_bounded}) in the
case that $V$ is bounded and $C_{d}=\mathsf{X}$. We will also invoke a density assumption which is weaker than the upper bound in (\ref{eq:minor_strong}).  It will be illustrated
through examples in Section \ref{sec:Examples} that (\ref{eq:intro_minor})-(\ref{eq:intro_drift})
can be satisfied in circumstances which allow $M$ to be non-ergodic.
Furthermore, it will also be demonstrated that, in contrast to (\ref{eq:U_bounded}),
conditions (\ref{eq:intro_minor})-(\ref{eq:intro_drift}) can be
satisfied with $U$ not bounded above, subject to strong enough assumptions
on $M$ and a restriction on the growth rate of the positive part
of $U$.

The main result obtained in the present work (Theorem \ref{thm:var_bound}
in Section \ref{sec:Non-asymptotic-variance}) is a bound of the form:
\begin{align*}
N> c_{1}\left(n+1\right)\geq \phi(x)& &\Longrightarrow& &\mathbb{E}_{x}^{N}\left[\left(\frac{\gamma_{n,x}^{N}\left(1\right)}{\gamma_{n,x}(1)}-1\right)^{2}\right]\leq c_{2}\frac{4}{N}\left(n+1\right)\frac{v^{2+\epsilon}(x)}{h_{0}^{2}(x)},
\end{align*}
with
\begin{equation*}
\phi(x):=c_1\left(\left\lceil\frac{1}{B_1}\log\left[B_0^2\frac{v(x)}{h_0(x)}\right]\right\rceil+1\right),
\end{equation*}
where $v(x)=e^{V(x)}$, $B_0$, $B_1$, $c_{1}$, $c_{2}$ are constants which are independent of $N$, $n$ and $x$ and for a real number $a$ we denote as $\left\lceil a \right\rceil$ the smallest integer $j$ such that $j\geq a$.
In this display $h_{0}$ is the eigenfunction associated with the principal eigenvalue of $Q$
and the constant $B_1$ is directly related to the size of the spectral gap of $Q$. 
Verification of the existence of $h_0$ along with various other spectral
quantities plays a central role in the proofs.

We note that \citet{smc:the:DMD04,smc:the:CdMG11} also consider the
case in which $\exp U(x)$ may touch zero and the former are also
directly concerned with approximation of the eigenvalue $\lambda$
corresponding to $h_{0}$ via the empirical probability measures $\left\{ \eta_{n}^{N}\right\} $.
These issues are beyond the scope of the present article but the study of these and related issues in a more general time-inhomogeneous
setting is underway. It is also remarked that \citet{smc:the:CdMG11}
consider a more general type of particle system, which involves an
accept/reject evolution mechanism. The approach taken here is also
applicable in that context, but for simplicity of presentation we
only consider the selection-mutation transition in \eqref{eq:particle_transitions}.

The remainder of the paper is structured as follows. Section \ref{sec:Multiplicative-ergodicity}
is largely expository: it introduces various spectral definitions
and the main assumptions of the present work and goes on to show how
these assumptions validate the application of multiplicative ergodicity
results of \citet{mc:the:KM05}. It is stressed that much of the content
of this section is included in order to make clear the similarities
and differences between the setting of interest and the main stated
assumptions and results of \citet{mc:the:KM05}. Section \ref{sec:Non-asymptotic-variance}
deals with the variance bounds for the particle approximations. Numerical
examples are given in Section \ref{sec:Examples}. Many of the proofs
of the results in Section \ref{sec:Multiplicative-ergodicity} are
in Appendix \ref{app:MET_proofs}. Some proofs and lemmas for the
results in Section \ref{sec:Non-asymptotic-variance} can be found
in Appendix \ref{app:variance}.

\section{Multiplicative Ergodicity\label{sec:Multiplicative-ergodicity} }

\subsection{Notations and Conventions}

Let $\mathsf{X}$ be a state space and $\mathcal{B}(\mathsf{X})$
be an associated countably generated $\sigma$-algebra. We are typically
interested in the case $\mathsf{X}=\mathbb{R}^{d_{x}}$, $d_{x}\geq1$,
but our results are readily applicable in the context of more general
non-compact state-spaces. For a weighting function $v:\mathsf{X}\rightarrow[1,\infty)$,
and $\varphi$ a measurable real-valued function on $\mathsf{X}$,
define the norm $\left\Vert \varphi\right\Vert _{v}:=\sup_{x\in\mathsf{X}}\left|\varphi(x)\right|/v(x)$
and let $\mathcal{L}_{v}:=\left\{ \varphi:\mathsf{X}\rightarrow\mathbb{R};\left\Vert \varphi\right\Vert _{v}<\infty\right\} $
be the corresponding Banach space. Throughout, when dealing with weighting
functions we employ an lower/upper-case convention for exponentiation
and write interchangeably $v\equiv e^{V}$. 

For $K$ a kernel on $\mathsf{X}\times\mathcal{B}\left(\mathsf{X}\right)$,
a function $\varphi$ and a measure $\mu$ denote $\mu(\varphi):=\int\varphi(x)\mu(dx)$,
$K\varphi(x):=\int K(x,dy)\varphi(y)$ and $\mu K(\cdot):=\int\mu(dx)K(x,\cdot)$.
Let $\mathcal{P}$ be the collection of probability measures on $\left(\mathsf{X},\mathcal{B}(\mathsf{X})\right)$,
and for a given weighting function $v:\mathsf{X}\rightarrow[1,\infty)$
let $\mathcal{P}_{v}$ denote the subset of such measures $\mu$ such
that $\mu(v)<\infty$. For $n\geq0$ the $n$-fold
iterate of $K$ is denoted: 
\begin{equation*}
K_{0}:=Id,\quad\quad K_{n}  :=  \underbrace{K\ldots K}_{n \text{ times}},\quad n\geq1.
\end{equation*}

The induced operator norm of a linear operator $K$ acting $\mathcal{L}_{v}\rightarrow\mathcal{L}_{v}$
is 
\[
\interleave K\interleave_{v}:=\sup\left\{ \frac{\left\Vert K\varphi\right\Vert _{v}}{\left\Vert \varphi\right\Vert _{v}};\varphi\in\mathcal{L}_{v},\left\Vert \varphi\right\Vert _{v}\neq0\right\} =\sup\left\{ \left\Vert K\varphi\right\Vert _{v};\varphi\in\mathcal{L}_{v},\left|\varphi\right|\leq v\right\} .
\]
The spectrum of $K$ as an operator on $\mathcal{L}_{v}$, denoted
by $\mathcal{S}_{v}(K)$, is the set of complex $z$ such that $\left[Iz-K\right]^{-1}$
does not exist as a bounded linear operator on $\mathcal{L}_{v}$.
The corresponding spectral radius of $K$, denoted by $\xi_{v}(K)$,
is given by 
\[
\xi_{v}(K):=\sup\left\{ \left|z\right|;z\in\mathcal{S}_{v}(K)\right\} =\lim_{n\rightarrow\infty}\interleave K_{n}\interleave_{v}^{1/n},
\]
 where the limit always exists by subadditive arguments, but may be
infinite. The following definitions are from \citet{mc:the:KM05}. 
\begin{itemize}
\item A pole $z_{0}\in\mathcal{S}_{v}(K)$ is \emph{of finite multiplicity}
$n$ if

\begin{itemize}
\item for some $\epsilon_{1}>0$ we have $\left\{ z\in\mathcal{S}_{v}(K);\left|z-z_{0}\right|\leq\epsilon_{1}\right\} =\left\{ z_{0}\right\}$, 
\item and the associated projection operator 
\[
J:=\frac{1}{2\pi i}\int_{\partial\left\{ z:\left|z-z_{0}\right|\leq\epsilon_{1}\right\} }\left[Iz-K\right]^{-1}\mathrm{d}z,
\]
 can be expressed as a finite linear combination of some $\left\{ s_{i}\right\} \subset\mathcal{L}_{v}$
and $\left\{ \nu_{i}\right\} \subset\mathcal{P}_{v}$, 
\[
J=\sum_{i,j=0}^{n-1}m_{i,j}\left[s_{i}\otimes\nu_{j}\right],
\]
 where $\left[s_{i}\otimes\nu_{j}\right](x,dy)=s_{i}(x)\nu_{j}(dy).$ 
\end{itemize}
\item $K$ admits a \emph{spectral gap }in $\mathcal{L}_{v}$ if there exists
$\epsilon_{0}>0$ such that $\mathcal{S}_{v}\left(K\right)\cap\left\{ z:\left|z\right|\geq\xi_{v}\left(K\right)-\epsilon_{0}\right\} $
is finite and contains only poles of finite multiplicity. 
\item $K$ is $v$\emph{-uniform }if it admits a spectral gap and there
exists a unique pole $\lambda\in\mathcal{S}_{v}\left(K\right)$ of
multiplicity 1, satisfying $\left|\lambda\right|=\xi_{v}\left(K\right)$. 
\item $K$ has a \emph{discrete spectrum} if for any compact set $B\subset\mathbb{C}\setminus\left\{ 0\right\} $,
$\mathcal{S}_{v}\left(K\right)\cap B$ is finite and contains only
poles of finite multiplicity. 
\item $K$ is \emph{$v$-separable }if for any $\epsilon>0$ there exists
a finite rank operator $\widehat{K}^{(\epsilon)}$ such that $\interleave K-\widehat{K}^{(\epsilon)}\interleave_{v}\leq\epsilon$ 
\end{itemize}

\subsection{Multiplicative Ergodic Theorem\label{sub:Multiplicative-ergodic-theorem}}

In this section we present the main assumptions and state some results
from \citet{mc:the:KM05} (see also \citet{mc:theory:KM03}).

\subsubsection{Assumptions}

\begin{condition} \label{hyp:irreducible}The semigroup $\left\{ Q_{n};n\geq1\right\} $
is $\psi$-irreducible and aperiodic (see \citet[Section 2.1]{mc:the:M06}).
\end{condition} \begin{condition} \label{hyp:drift}There exists
an unbounded $V:\mathsf{X}\rightarrow[1,\infty)$, constants $m_{0}\geq1$,
$\delta\in(0,1)$ and $\underline{d}\geq1$ with the following properties:

For each $d\geq\underline{d}$ and $C_d:=\{x\in\mathsf{X};V(x)\leq d\}$, 

\quad $\bullet$ there exists $\epsilon_{d}\in (0,1]$ and
$\nu_{d}\in\mathcal{P}_{v}$ such that $C_{d}$
is $\left(m_{0},\epsilon_{d},\nu_{d}\right)$-small for $Q$, i.e.,
\begin{eqnarray}
Q_{m_{0}}(x,\cdot) & \geq & \mathbb{I}_{C_{d}}(x)\epsilon_{d}\nu_{d}(\cdot),\quad\forall x\in\mathsf{X},\label{eq:A_minor}
\end{eqnarray}
\quad \quad with $\nu_{d}\left(C_{d}\right)>0$. Furthermore $Q_{m_{0}}\left(C_{d}\right)(x)>0$
for all $x\in\mathsf{X}$. 

\quad $\bullet$ there exists $b_{d}<\infty$ such
that the following multiplicative drift condition holds, 
\begin{eqnarray}
Q\left(e^{V}\right) & \leq & e^{V\left(1-\delta\right)+b_{d}\mathbb{I}_{C_{d}}}.\label{eq:A_drift}
\end{eqnarray}
 \end{condition} \begin{condition} \label{hyp:U+_in_V} $U:\mathsf{X}\rightarrow\mathbb{R}$
is such that 
\begin{eqnarray*}
U^{+} & := & \max\left(U,0\right)\in\mathcal{L}_{V}.
\end{eqnarray*}
 \end{condition} \begin{condition} \label{hyp:density}There exists
$t_{0}\geq1$ and for each $d\geq\underline{d}$ there exists a measure
$\beta_{d}$, such that $\beta_{d}\left(e^{V}\right)<\infty$ and
\begin{eqnarray*}
\mathbb{P}_{x}\left(X_{t_{0}}\in A,\tau_{C_{d}^{c}}>t_{0}\right) & \leq & \beta_{d}\left(A\right),\quad x\in C_{d},A\in\mathcal{B}\mathsf{\left(X\right)},
\end{eqnarray*}
 where $\mathbb{P}_{x}$ denotes the law of the Markov chain $\{X_{n}\}$
with transition $M$ and $\tau_{A}:=\inf\left\{ n\geq1:X_{n}\in A\right\} $.\end{condition}
\begin{rem} We take care to emphasize the following differences and
similarities between the above assumptions and the setting of \citet{mc:the:KM05}.\end{rem} 
\begin{itemize}
\item Assumption (H\ref{hyp:drift}) equation \eqref{eq:A_drift} applies
directly to the $Q$ kernel, whereas \citet{mc:the:KM05} impose a
multiplicative drift condition on $M$. The key issue is that the
multiplicative drift condition for $Q$ is the essential and implicit
ingredient of Lemma B.4 of \citet{mc:the:KM05}, and as we shall see
in Section \ref{sec:Examples}, under the conditions that $U$ is
bounded above but not bounded below, assumption (H\ref{hyp:drift})
can hold without geometric drift assumptions on $M$. A related phenomenon
is considered by \citet{mc:the:M06} in order to obtain {}``one-sided''
large deviation principles for ergodic sample-path averages for the
chain with transition $M$. 
\item Assumption (H\ref{hyp:drift}) requires the sublevel sets of $V$
to be small for $Q$ and this is exploited in Lemma \ref{lem:P_check_M-drift}.
The explicit $m_{0}$-step minorisation condition makes it easy
to bound below the spectral radius of $Q$, see Lemma \ref{lem:spr_bounded_below}.
In the setting of \citet{mc:theory:KM03} the spectral radius of $Q$ is bounded below
by $1$ as $U$ is assumed centered with respect to the invariant
probability distribution for $M$. In the present context, this centering
assumption is unnatural, especially as we want to consider some situations
where such an invariant probability does not exist. 
\item Assumption (H\ref{hyp:U+_in_V}) is weaker than the corresponding
assumption in the statement of \citep[ Theorem 3.1]{mc:the:KM05}.
However, (H\ref{hyp:U+_in_V}) coincides with the first part of \citep[Equation 73]{mc:the:KM05},
which combined with (H\ref{hyp:irreducible}), (H\ref{hyp:drift})
and (H\ref{hyp:density}) in Lemma \ref{lem:Q_sep} below, is enough
to prove that $Q$ has a discrete spectrum in $\mathcal{L}_{v}$. 
\item As shown in \citep[Theorem 3.4]{mc:the:KM05} and \citep{mc:theory:KM03},
a$ $ MET can be proved without (H\ref{hyp:density}), but at the
cost of restrictions on the class of functions to which $U$ belongs
which are a little unwieldy. 
\end{itemize}

\subsubsection{Results}

We now give a collection of results which are used to prove the MET, Theorem \ref{thm:MET}.  
The proofs are given in Appendix \ref{app:MET_proofs}. It is remarked
that the steps in the proof of Theorem \ref{thm:MET} are effectively the
same as part of the proof of Theorem 3.1 of \citet{mc:the:KM05}, however, our starting assumptions
are stated differently.

The following preparatory lemma establishes that the Feynman-Kac formula
(\ref{eq:FK_formula}) is well defined and presents bounds on the
spectral radius of $Q$. \begin{lem} \label{lem:spr_bounded_below}
Assume (H\ref{hyp:drift}). Then for all $x\in\mathsf{X},\; n\geq1,\;\varphi\in\mathcal{L}_{v}$,
\begin{eqnarray}
\left|\gamma_{n,x}\left(\varphi\right)\right|<\infty,\label{eq:FK_well-defined}
\end{eqnarray}
 and for all $d\geq\underline{d}$, 
\begin{eqnarray}
\epsilon_{d}\nu_{d}\left(C_{d}\right)\leq\xi_{v}\left(Q\right)<\infty,\label{eq:spr_lowerbound}
\end{eqnarray}
 where $\underline{d}$ is as in (H\ref{hyp:drift}).\end{lem}

To clarify how assumptions (H\ref{hyp:irreducible})-(H\ref{hyp:density})
connect with the results of \citet{mc:the:KM05} we next present a
lemma regarding the $v$-separability of $Q$ which is a stepping
stone to the MET. Observe that the multiplicative drift condition
(H\ref{hyp:drift}) implies that $Q$ can be approximated in norm
to arbitrary precision by truncation to the sublevel sets of $V$,
in the sense that for any $r\geq \underline{d}$, 
\begin{eqnarray}
\mathbb{I}_{C_{r}^{c}}Q\left(e^{V}\right) & \leq & e^{V-\delta r},\label{eq:approx_Q}
\end{eqnarray}
 and then with $\widehat{Q}^{(r)}:=\mathbb{I}_{C_{r}}Q$, it follows immediately
that $\interleave Q-\widehat{Q}^{(r)}\interleave_{v}\leq e^{-\delta r}$.
In the following lemma, which combines \citep[Lemmata B.3-B.5]{mc:the:KM05}
and is included here for completeness, the density assumption (H\ref{hyp:density})
plays a key role in establishing that iterates of this truncation
of $Q$ can be approximated by a finite rank kernel.  \begin{lem}
\label{lem:Q_sep}Assume (H\ref{hyp:irreducible})-(H\ref{hyp:density}).
Then $Q_{2t_{0}+2}$ is $v$-separable, where $t_{0}$ is as in (H\ref{hyp:density}).\end{lem}

The following theorem makes a key connection between $v$-separability
and a discrete spectrum. \begin{thm} \citep[Theorem 3.5]{mc:the:KM05}
If the linear operator $Q:\mathcal{L}_{v}\rightarrow\mathcal{L}_{v}$
is bounded and $Q_{t_{0}}:\mathcal{L}_{v}\rightarrow\mathcal{L}_{v}$
is $v$-separable for some $t_{0}\geq1$, then $Q$ has a discrete
spectrum in $\mathcal{L}_{v}$. \end{thm} Under (H\ref{hyp:drift})
$Q$ is indeed bounded, so has a discrete spectrum in $\mathcal{L}_{v}$
and then by definition it also admits a spectral gap in $\mathcal{L}_{v}$.
For any $\theta>\xi_{v}\left(Q\right)$ we may consider the resolvent operator
defined by 
\begin{equation}
R_\theta  :=  \left[I\theta-Q\right]^{-1}=\sum_{k=0}^{\infty}\theta^{-k-1}Q_{k},\label{eq:R_defn}\\
\end{equation}

We can now state and prove the MET: \begin{thm} \label{thm:MET}
Assume (H\ref{hyp:irreducible})-(H\ref{hyp:density}). Then $\lambda=\xi_{v}\left(Q\right)$ is a maximal and isolated eigenvalue for $Q$. For any $d\geq\underline{d}$ and $\theta>\xi_v(Q)$,  the operator $H_{\theta,d}$ defined by
\begin{equation}
H_{\theta,d} :=  \left[I\lambda_{\theta}-\left(R_\theta-\theta^{\left(-m_{0}-1\right)}\epsilon_{d}\mathbb{I}_{C_{d}}\otimes\nu_{d}\right)\right]^{-1}=\sum_{k=0}^{\infty}\lambda_{\theta}^{-k-1}\left(R_\theta-\theta^{\left(-m_{0}-1\right)}\epsilon_{d}\mathbb{I}_{C_{d}}\otimes\nu_{d}\right)^{k},\label{eq:H_defn}
\end{equation}
is bounded as an operator on $\mathcal{L}_v$, with $\lambda_{\theta}:=(\theta-\xi_{v}\left(Q\right))^{-1}$.   

The function $h_0\in\mathcal{L}_v$ and measure $\mu_0\in\mathcal{P}_v$ defined by
\begin{align}
h_{0}:=\frac{H_{\theta,d}\left(\mathbb{I}_{C_{d}}\right)}{\mu_{0}H_{\theta,d}\left(\mathbb{I}_{C_{d}}\right)}, &  & \mu_{0}:=\frac{\nu_{d}H_{\theta,d}}{\nu_{d}H_{\theta,d}\left(1\right)}.\label{eq:h_0_mu_0_defn}
\end{align}
are independent of $\theta,d$ and satisfy 
\begin{align*}
Qh_{0}=\lambda h_{0}, &  & \mu_{0}Q=\lambda\mu_{0},& &\mu_{0}\left(h_{0}\right)=1.
\end{align*}

Furthermore, there exist constants $B_{0}<\infty$ and $B_{1}>0$
such that for any $\varphi\in\mathcal{L}_{v}$, any $n\geq1$ and any $x\in\mathsf{X}$,
\begin{eqnarray}
\left|\lambda^{-n}\gamma_{n,x}(\varphi)-h_{0}(x)\mu_{0}(\varphi)\right| & \leq & \left\Vert \varphi\right\Vert _{v}B_{0}e^{-nB_{1}}v(x).\label{eq:MET}
\end{eqnarray}

\end{thm}

\begin{proof} We give only a sketch proof, as it is essentially that
of Theorem 3.1 of \citet{mc:the:KM05}. As established in Lemma \ref{lem:spr_bounded_below},
under our assumptions $0<\xi_{v}\left(Q\right)<\infty$. Furthermore the semigroup associated with $Q$ is $\psi$-irreducible, and as observed above $Q$ is bounded on $\mathcal{L}_{v}$, has a discrete
spectrum and therefore admits a spectral gap in $\mathcal{L}_{v}$. Proposition 2.8 of \citet{mc:the:KM05} therefore applies. Thus $Q$ is $v$-uniform and $\lambda=\xi_v(Q)$ is a maximal and isolated eigenvalue.

By the minorization condition of (H\ref{hyp:drift}) one can obtain a minorization condition
for $R_\theta$ of \eqref{eq:R_defn}: 
\[
R_\theta(x,dy)\geq\theta^{\left(-m_{0}-1\right)}\epsilon_{d}\mathbb{I}_{C_{d}}(x)\nu_{d}(dy),
\]
which holds for any $d\geq\underline{d},\theta>\xi_v(Q)$. Therefore by the argument in \citet{mc:the:KM05}[Proof of Proposition 2.8], for any $\theta>\xi_v(Q)$ and $d\geq \underline{d}$, the spectral radiue of $[R_\theta - \theta^{\left(-m_{0}-1\right)}\epsilon_{d}\mathbb{I}_{C_d}\otimes\nu_d]$ is strictly less than $\lambda_\theta=(\theta-\xi_v(Q))^{-1}$. Thus $H_{\theta,d}$ is bounded as an operator on $\mathcal{L}_v$ and the sum in \eqref{eq:H_defn} converges in the operator norm.  

Then also by \citep{mc:the:KM05}[Proposition 2.8], $H_{\theta,d}(\mathbb{I}_{C_d})\in\mathcal{L}_{v}$
is an eigenfunction for $Q$ with eigenvalue $\lambda=\xi_{v}\left(Q\right)$. By similar arguments
to \citet{mc:theory:KM03}[proof of Proposition 4.5] it is easily verfied that $\nu_d H_{\theta,d}$ is an eigenmeasure. The normalization to $h_0$ and $\mu_0$
is justified by the finiteness, under our assumptions, of the associated
quantities. By \citep{mc:theory:KM03}[Theorem 3.3 part (iii), see also comments on p.332] $h_0$ and $\mu_0$ constructed using any $\theta,d$ are respectivaly the $\psi$-essentially unique eigenfunction and unique eigenmeasure satisfying $\mu_0(\mathsf{X})=1$, $\mu_0(h_0)=1$, hence the lack of dependence on $\theta,d$.

To obtain \eqref{eq:MET} one may define the \emph{twisted} kernel:
\begin{eqnarray}
\check{P}(x,dy) & := & \lambda^{-1}h_{0}^{-1}(x)Q(x,dy)h_{0}(y),\label{eq:twisted_kernel_defn}
\end{eqnarray}
 which can be seen to be well defined as a Markov kernel, as $\lambda$
is strictly positive and finite and (H\ref{hyp:drift}) implies $h_{0}$
is everywhere finite and strictly positive. Furthermore one observes
immediately that $\check{P}$ admits $\check{\pi}$, defined by $\check{\pi}(\varphi)=\mbox{\ensuremath{\mu}}_{0}\left(h_{0}\varphi\right)/\mbox{\ensuremath{\mu}}_{0}\left(h_{0}\right)=\mbox{\ensuremath{\mu}}_{0}\left(h_{0}\varphi\right)$,
as an invariant probability distribution. By Lemma \ref{lem:P_check_M-drift}
in Appendix \ref{app:MET_proofs} one can apply Theorem 3.4 of \citet{mc:the:KM05}
to the Markov chain associated to the twisted kernel, (in the notation
of of Theorem 3.4 of \citet{mc:the:KM05}, take $g\equiv\varphi/h_{0}$,
$F\equiv0$). This results in the bound \eqref{eq:MET}, which completes
the proof. \end{proof}

\begin{rem} Upon dividing through by $h_{0}$, the equation (\ref{eq:MET}) of the MET
may be viewed as a probabilistic, geometric ergodic theorem for
the twisted chain associated to the kernel \eqref{eq:twisted_kernel_defn}
and the modified test function $\varphi/h_{0}$, with a naturally
modified drift function $\check{v}=e^{\check{V}}$ proportional to
${v}/{h_{0}}$. See Lemma \ref{lem:P_check_M-drift}
in Appendix \ref{app:MET_proofs}. \end{rem}

\begin{rem}
The constant $B_1$ in equation \eqref{eq:MET} is directly related to the size of the spectral gap of $Q$, see \cite[Proof of Theorem 4.1]{mc:theory:KM03}.
\end{rem}

\section{Non-Asymptotic Variance\label{sec:Non-asymptotic-variance}}

\subsection{Tensor Product Functionals}

The various tensor product functionals considered in the remainder
of this paper require some additional notation. For a measurable function
$F$ on $\mathsf{X}^{2}$ and a weighting function $v:\mathsf{X}\rightarrow[1,\infty)$,
we define the norm $\left\Vert F\right\Vert _{v,2}:=\sup_{x,y\in\mathsf{X^{2}}}\left|F(x,y)\right|/\left(v(x)v(y)\right)$
and denote $\mathcal{L}_{v,2}:=\left\{ F:\mathsf{X}^{2}\rightarrow\mathbb{R};\left\Vert F\right\Vert _{v,2}<\infty\right\} $
the corresponding function space. For two functions $\varphi_{1},\varphi_{2}\in\mathcal{L}_{v}$,
we denote by $\varphi_{1}\otimes\varphi_{2}\in\mathcal{L}_{v,2}$
the tensor product function defined by $\varphi_{1}\otimes\varphi_{2}(x,x'):=\varphi_{1}(x)\varphi_{2}(x')$.
Let $K:\mathsf{X}\times\mathcal{B}(\mathsf{X})\rightarrow\mathbb{R}_{+}$
be a kernel on $\mathsf{X}$. The two-fold tensor product operator
corresponding to $K$ is defined, for any $F\in\mathcal{L}_{v,2}$,
by 
\begin{eqnarray*}
K^{\otimes2}\left(F\right)(x,x') & := & \int_{\mathsf{X}^{2}}K(x,dy)K(x',dy')F(y,y').
\end{eqnarray*}
The iterated operator notation of the previous section is carried
over so that 
\begin{equation*}
K_{0}^{\otimes2}:=Id,\quad\quad K_{n}^{\otimes2} := \underbrace{K^{\otimes2}\ldots K^{\otimes2}}_{n\text{ times}},\quad n\geq1.
\end{equation*}

Corresponding to the particle empirical measures of section \ref{sub:Particle-systems},
for $n\geq1$, we introduce the tensor product empirical measures
(or 2-fold $V-$statistic): 
\[
\left(\eta_{n}^{N}\right)^{\otimes2}:=\frac{1}{N^{2}}\sum_{1\leq i,j\leq N}\delta_{\left(\zeta_{n}^{i},\zeta_{n}^{j}\right)},\quad\left(\gamma_{n,x}^{N}\right)^{\otimes2}:=\gamma_{n,x}^{N}(1)^{2}\left(\eta_{n}^{N}\right)^{\otimes2}.
\]

Following the definition of \citet{smc:the:CdMG11}, the coalescent
integral operator $D$, acting on functions on $\mathsf{X}^{2}$,
is defined by 
\[
D\left(F\right)\left(x,x'\right)=F\left(x,x\right),\quad\left(x,x'\right)\in\mathsf{X}^{2}.
\]
 For any $0\leq s\leq\left(n+1\right)$, we denote by $\mathcal{I}_{n,s}:=\left\{ (i_{1},...,i_{s})\in\mathbb{N}_{0}^{s};0\leq i_{1}<\ldots<i_{s}\leq n\right\} $
the set of coalescent time configurations over a horizon of length
$n+1$ and for $(i_{1},...,i_{s})\in\mathcal{I}_{n,s}$ and $x\in\mathsf{X}$,
the nonegative measure $\Gamma_{n,x}^{(i_{1},...i_{s})}$ on $\left(\mathsf{X}^{2},\mathcal{B}\mathsf{\left(X^{2}\right)}\right)$,
and its normalised counterpart $\bar{\Gamma}_{n,x}^{(i_{1},...,i_{s})}$,
are defined by
\begin{eqnarray}
\Gamma_{n,x}^{(i_{1},...,i_{s})}:=\gamma_{i_{1},x}^{\otimes2}DQ_{i_{2}-i_{1}}^{\otimes2}D\ldots Q_{i_{s}-i_{s-1}}^{\otimes2}DQ_{n-i_{s}}^{\otimes2}, & \quad & \bar{\Gamma}_{n,x}^{(i_{1},...i_{s})}:=\frac{\Gamma_{n,x}^{(i_{1},...,i_{s})}}{\gamma_{n,x}(1)^{2}},\label{eq:Gamma_defn}
\end{eqnarray}
 for $s\geq1$, and for $s=0$, $\Gamma_{n,x}^{\left(\emptyset\right)}\left(F\right):=\gamma_{n,x}^{\otimes2}\left(F\right)$
and $\bar{\Gamma}_{n,x}^{\left(\emptyset\right)}\left(F\right):=\eta_{n}^{\otimes2}\left(F\right)$.
We refer the reader to \citet[ Section 3]{smc:the:CdMG11} for a
helpful visual representation of the integrals in the transport equation
(\ref{eq:Gamma_defn}). We have already checked in Lemma \ref{lem:spr_bounded_below}
that the Feynman-Kac formula (\ref{eq:FK_formula}) is well defined
under our assumptions in the $\mathcal{L}_{v}$ setting, which validates
the denominator of (\ref{eq:Gamma_defn}).

When Theorem \ref{thm:MET} holds, we will denote by $\check{\mathbb{E}}_{x}$
expectation with respect to the law of the twisted Markov chain $\left\{ \check{X}_{n};n\geq0\right\} $,
i.e that with transition kernel $\check{P}$ as in equation (\ref{eq:twisted_kernel_defn})
and initialised from $\check{X}_{0}=x$.

\subsection{Non-Asymptotic Variance}

In this section we give our main result. The proof is detailed in
section \ref{sec:const_proof}. The following additional assumption
imposes some further restrictions on the function class considered,
but this is not overly demanding, considering that we will be dealing
with coalesced tensor product quantities. \begin{condition} \label{hyp_drift_var}Let
$V$ and $\bar{d}$ be as in assumption (H\ref{hyp:drift}). There
exists $0<\epsilon_{0}<\epsilon$ and for all $d\geq\bar{d}$, there
exists $b_{d}^{*}<\infty$ such that 
\[
Q\left(e^{(1+\epsilon)V}\right)\leq e^{(1+\epsilon)V-(1+\epsilon_{0})V+b_{d}^{*}\mathbb{I}_{C_{d}}}.
\]
\end{condition}
The following theorem is due to \citet{smc:the:CdMG11}. 

\begin{thm} \label{thm:cerou}\citep[Proposition 3.4]{smc:the:CdMG11}
For any $n\geq1$, $x\in\mathsf{X}$ and $N\geq1$ the following expansion
holds: 
\begin{eqnarray}
 &  & \mathbb{E}_{x}^{N}\left[\left(\frac{\gamma_{n,x}^{N}\left(1\right)}{\gamma_{n,x}(1)}-1\right)^{2}\right]\nonumber \\
 &  & =\sum_{s=1}^{n+1}\left(1-\frac{1}{N}\right)^{\left(n+1\right)-s}\frac{1}{N^{s}}\sum_{\left(i_{1},...,i_{s}\right)\in\mathcal{I}_{n,s}}\left[\bar{\Gamma}_{n,x}^{(i_{1},...i_{s})}\left(1\otimes1\right)-1\right],\label{eq:l2_delmoral_decomp}
\end{eqnarray}
 where $\mathbb{E}_{x}^{N}$ denotes expectation w.r.t. the law of
the $N$-particle system. \end{thm}
A full proof is not provided here. However, we note that we may write
\begin{equation}
\mathbb{E}_{x}^{N}\left[\left(\frac{\gamma_{n,x}^{N}\left(1\right)}{\gamma_{n,x}(1)}-1\right)^{2}\right]=\frac{\mathbb{E}_{x}^{N}\left[\left(\gamma_{n,x}^{N}\right)^{\otimes 2}\left(1\otimes 1\right)\right]}{\gamma_{n,x}(1)^{2}}-1,\label{eq:vareq}
\end{equation}
where the equality is due to the lack of bias property $\mathbb{E}_{x}^{N}\left[\gamma_{n,x}^N(1)\right]=\gamma_{n,x}(1)$ and the definition of $\left(\gamma_{n,x}^{N}\right)^{\otimes 2}$. In summary, the proof of Theorem \ref{thm:cerou} involves  recursive calculation of the expectation on the right of  \eqref{eq:vareq}, followed by organisation of the resulting terms into the form \eqref{eq:l2_delmoral_decomp}. The reader is directed to \citep{smc:the:CdMG11} for the details.

It is remarked that there is a different error decomposition in \citep{chan:lai}, which can hold
to any order under appropriate regularity conditions; one would conjecture that
this decomposition can also be treated, but this is not considered
here. The main result of this section is the following theorem, whose
proof is postponed. 
\begin{thm} \label{thm:var_bound}Assume (H\ref{hyp:irreducible})-(H\ref{hyp_drift_var}).
Then there exists $c_{1}<\infty$ and $c_{2}<\infty$ depending only
on the quantities in (H\ref{hyp:irreducible})-(H\ref{hyp_drift_var})
such that for all $x\in\mathsf{X}$, 
\begin{align*}
N> c_{1}\left(n+1\right)\geq \phi(x)&  &\Longrightarrow&  &\mathbb{E}_{x}^{N}\left[\left(\frac{\gamma_{n,x}^{N}\left(1\right)}{\gamma_{n,x}(1)}-1\right)^{2}\right]\leq c_{2}\frac{4}{N}\left(n+1\right)\frac{v^{2+\epsilon}(x)}{h_{0}^{2}(x)},
\end{align*}
with
\begin{equation*}
\phi(x):=c_1\left(\left\lceil\frac{1}{B_1}\log\left[B_0^2\frac{v(x)}{h_0(x)}\right]\right\rceil+1\right),
\end{equation*}
and where $B_0$ and $B_1$ are as in Theorem \ref{thm:MET}.
\end{thm}

\subsection{Construction of the Proof}

\label{sec:const_proof}

In the following Section, we detail the argument to prove Theorem
\ref{thm:var_bound}. To that end, we present the essence of the argument
with Proposition \ref{prop:Gamma_bound} and Lemma \ref{lem:E_check_prod_bounds}
below; the proofs of which are in Appendix \ref{app:variance} along with some supporting results.

The proof of Theorem \ref{thm:var_bound} is constructed in the following
manner. By Theorem \ref{thm:cerou} we have the decomposition \eqref{eq:l2_delmoral_decomp}
in terms of the operators $\left\{ \bar{\Gamma}_{n,x}^{(i_{1},...i_{s})}\right\} $.
The proof in \citet{smc:the:CdMG11} focuses upon controlling these
expressions via the regularity conditions mentioned in section \ref{sec:standard_reg_assumptions};
our proof will do the same, except under (H\ref{hyp:irreducible})-(H\ref{hyp_drift_var}).

Throughout the remainder of this paper, let $V^{*}:\mathsf{X}\rightarrow[1,\infty)$
is defined by 
\begin{eqnarray}
V^{*}(x) & := & V(x)\left(1+\epsilon\right)-\log h_{0}(x)+\log\left\Vert h_{0}\right\Vert _{v^{\left(1+\epsilon\right)}},\label{eq:V_star_defn}
\end{eqnarray}
where $\epsilon$ is as in (H\ref{hyp_drift_var}). We proceed with the
following key proposition.
\pagebreak
\begin{prop} \label{prop:Gamma_bound}Assume (H\ref{hyp:irreducible})-(H\ref{hyp_drift_var}).
Then there exists $c<\infty$ depending only on the quantities in (H\ref{hyp:irreducible})-(H\ref{hyp_drift_var})
such that for all $n\geq1,$ $0\leq s\leq n+1$, $(i_{1},...i_{s})\in\mathcal{I}_{n,s}$,
$F\in\mathcal{L}_{v^{1/2},2}$ and $x\in\mathsf{X}$, 
\begin{eqnarray}
\bar{\Gamma}_{n,x}^{(i_{1},...i_{s})}\left(F\right) & \leq & \left\Vert F\right\Vert _{v^{1/2},2}c^{s+1}\frac{v(x)}{h_{0}(x)}\frac{\check{\mathbb{E}}_{x}\left[\prod_{k\in\left\{ i_{1},\ldots,i_{s-1}\right\} }v\left(\check{X}_{k}\right)v^{*}\left(\check{X}_{i_{s}}\right)\right]}{\check{\mathbb{E}}_{x}\left[1/h_{0}\left(\check{X}_{n}\right)\right]^{2}},\label{eq:twisted_E_bound}
\end{eqnarray}
 with the conventions that the product in the numerator is unity when
$s\leq1$, and in the case of $s=0,$ $i_{s}=0$. In the above display,
$v$ is as in (H\ref{hyp:drift}), $h_{0}\in\mathcal{L}_{v}$
is the eigenfunction as in Theorem \ref{thm:MET}
and $v^{*}=e^{V^{*}}$ is as in \eqref{eq:V_star_defn}. \end{prop}

This result of Proposition \ref{prop:Gamma_bound} connects the operators
$\left\{ \bar{\Gamma}_{n,x}^{(i_{1},...i_{s})}\right\} $ with expectations
of the Lyapunov functions $v$ and $v^{*}$ and the eigenfunction,
w.r.t.~the twisted chain. Given this result, one needs to control
the numerator and denominator. The latter can be achieved by the MET
of Theorem \ref{thm:MET} and the former via the following:

\begin{lem} \label{lem:E_check_prod_bounds}Assume (H\ref{hyp:irreducible})-(H\ref{hyp_drift_var}).
Then there exists $c<\infty$ depending only on the quantities in (H\ref{hyp:irreducible})-(H\ref{hyp_drift_var})
such that for any $n\geq1$, $1\leq s\leq n+1$, $(i_{1},...,i_{s})\in\mathcal{I}_{n,s}$,
\begin{eqnarray}
\check{\mathbb{E}}_{x}\left[\prod_{k\in\left\{ i_{1},\ldots,i_{s}\right\} }v\left(\check{X}_{k}\right)v^{*}\left(\check{X}_{n+1}\right)\right] & \leq & c^{s+1}v^{*}(x),\quad\forall x\in\mathsf{X},\label{eq:E_check_bounded_cs}
\end{eqnarray}
 where $v^{*}$ is as in \eqref{eq:V_star_defn}. \end{lem}
We now proceed with the proof of Theorem \ref{thm:var_bound}.

\begin{proof}{[}Proof of Theorem \ref{thm:var_bound}{]} By Proposition
\ref{prop:Gamma_bound} and Lemma \ref{lem:E_check_prod_bounds} we
have that there exists a finite constant $c$ depending only on the
quantities in (H\ref{hyp:irreducible})-(H\ref{hyp_drift_var}) such
that 
\begin{eqnarray}
\bar{\Gamma}_{n,x}^{(i_{1},...i_{s})}\left(1\otimes1\right) & \leq & c^{s+1}\frac{v(x)}{h_{0}(x)}v^{*}(x)\frac{1}{\check{\mathbb{E}}_{x}\left[1/h_{0}\left(\check{X}_{n}\right)\right]^{2}}.\label{eq:thm_var_Gamma_bound1}
\end{eqnarray}
Using the fact that $\check{\mathbb{E}}_{x}\left[1/h_{0}\left(\check{X}_{n}\right)\right]=\gamma_{n,x}(1)/[\lambda^{n}h_{0}(x)]$
we appeal to \eqref{eq:MET} of the MET of Theorem \ref{thm:MET} as follows.
Without loss of generality, it can be assumed that $B_0>1$. Then for all $x\in\mathsf{X}$ 
\begin{align}
n\geq\left\lceil\frac{1}{B_1}\log\left[B_0^2\frac{v(x)}{h_0(x)}\right]\right\rceil && \Rightarrow && 1-B_0e^{-B_1 n}\frac{v(x)}{h_0(x)}\geq \frac{B_0-1}{B_0}&&\Rightarrow&&\check{\mathbb{E}}_{x}\left[1/h_{0}\left(\check{X}_{n}\right)\right]\geq \frac{B_0-1}{B_0}.\label{eq:phi_pre_defn}
\end{align}
Throughout the remainder of the proof the left-most inequality in \eqref{eq:phi_pre_defn} is assumed to hold. Then combining \eqref{eq:phi_pre_defn} with (\ref{eq:thm_var_Gamma_bound1}) and recalling
the definition of $v^{*}$ we have that there exists $c_{0}<\infty$
such that 
\begin{eqnarray*}
\bar{\Gamma}_{n,x}^{(i_{1},...i_{s})}\left(1\otimes1\right) & \leq & c_{0}c^{s+1}\frac{v^{2+\epsilon}(x)}{h_{0}^{2}(x)}.
\end{eqnarray*}
Proceeding by the essentially the same argument as in \citep[Proof of Theorem 5.1]{smc:the:CdMG11},
we use the identity: 
\[
\sum_{s=1}^{n+1}\sum_{\left(i_{1},...,i_{s}\right)\in\mathcal{I}_{n,s}}\prod_{j\in\left\{ i_{1},...,i_{s}\right\} }a_{j}=\left[\prod_{s=0}^{n}\left(1+a_{s}\right)\right]-1,
\]
which holds for any $n\geq1$ and $\left\{ a_{s};s\geq0\right\} $, to establish
via Theorem \ref{thm:cerou} that 
\begin{eqnarray*}
\mathbb{E}_{x}^{N}\left[\left(\frac{\gamma_{n,x}^{N}\left(1\right)}{\gamma_{n,x}(1)}-1\right)^{2}\right] & \leq & c_{0}c\frac{v^{2+\epsilon}(x)}{h_{0}^{2}(x)}\sum_{s=1}^{n+1}\left(1-\frac{1}{N}\right)^{\left(n+1\right)-s}\frac{1}{N^{s}}\sum_{\left(i_{1},...,i_{s}\right)\in\mathcal{I}_{n,s}}c^{s}\\
 & = & c_{0}c\frac{v^{2+\epsilon}(x)}{h_{0}^{2}(x)}\left(1-\frac{1}{N}\right)^{n+1}\left[\left(1+\frac{c}{N-1}\right)^{n+1}-1\right]\\
 & \leq & c_{0}c\frac{v^{2+\epsilon}(x)}{h_{0}^{2}(x)}\left[\left(1+\frac{c}{N-1}\right)^{n+1}-1\right].
\end{eqnarray*}
Then exactly as in \citep[Proof of Corollary 5.2]{smc:the:CdMG11},
\begin{align*}
N>1+c\left(n+1\right)&&\Rightarrow&&\left(1+\dfrac{c}{N-1}\right)^{n+1}-1  \leq\frac{2}{N-1}c\left(n+1\right)
\leq\frac{4}{N}c\left(n+1\right).
\end{align*}
This completes the proof. \end{proof}

\section{Examples\label{sec:Examples}}

This section gives some discussion and examples of circumstances in
which the assumptions can be satisfied. In particular we focus on
the drift assumption of (H\ref{hyp:drift}). It seems natural to consider
two general cases: those in which it is not assumed, or it is assumed, that the
Markov kernel $M$ itself satisfies a multiplicative drift condition.

\subsection{Cases without a multiplicative drift assumption on $M$}

In this situation, the decay of the potential function plays a key
role in establishing the multiplicative drift condition, illustrated
as follows.

\begin{lem} Assume that there exists $V:\mathsf{X}\rightarrow[1,\infty)$
unbounded such that $\interleave M\interleave_{v}<\infty$ and for all $d\geq1$, $C_d$ is $(1,\epsilon_d,\nu_d)$-small for $M$, with $\nu_d(C_d)>0$ and $M(C_d)(x)>0$ for all $x$.  If for all $d\geq1$, $\inf_{x\in C_{d}}U(x)>-\infty$,
and there exists $d_{1}$ such that $\sup_{x\in C_{d_{1}}}U(x)<\infty$
and for some $\delta_{1}\in(0,1)$, $\sup_{x\in C_{d_{1}}^{c}}U(x)/V(x)\leq-\delta_{1}$,
assumption (H\ref{hyp:drift}) is satisfied. \end{lem} \begin{proof}
We have 
\[
Q\left(e^{V}\right)(x)\leq\exp\left(V(x)+U(x)+\log\interleave M\interleave_{v}\right),\quad\forall x\in\mathsf{X}.
\]
As $V$ is unbounded, for any $\delta\in(0,\delta_{1})$ there exists
$\underline{d}$ large enough such that for all $x\in\mathsf{X}$
and $d\geq\underline{d}$, 
\[
\mathbb{I}_{C_{d}^{c}}(x)Q\left(e^{V}\right)(x)\leq\exp\left(V(x)(1-\delta)\right),\quad\mathbb{I}_{C_{d}}(x)Q\left(e^{V}\right)(x)\leq\exp\left(d+\sup_{y\in C_{d}}U(y)+\log\interleave M\interleave_{v}\right),
\]
which is enough to verify the drift part of (A2). The minorization condition with $m_0=1$ and the $Q(C_d)(x)>0$ part are direct as $U(x)$ is bounded below on $C_{d}$. \end{proof} 

In the extensive literature on Lyapunov drift for Markov kernels there are several conditions which immediately guarantee the existence of $v$ such that $\interleave M\interleave_{v}<\infty$. For example, any $M$ satisfying the polynomial drift condition of \cite{jarner2002polynomial} automatically satisfies $\interleave M\interleave_{v}<\infty$ for the same $v$ up to a factor of $e$.  However, ergodicity of $M$ is not necessary, as illustrated in the following simple example.

\subsubsection{Gaussian Random Walk}

Let $\mathsf{X}:=\mathbb{R}$ and $U$ and $M$ be defined by 
\[
U(x):=-x^{2},\quad M(x,dy):=\dfrac{1}{\sqrt{2\pi}}\exp\left(-\dfrac{\left(y-x\right)^{2}}{2}\right)dy,
\]
 where $dy$ denotes Lebesgue measure. Taking $\psi$ as Lebesgue measure,
the $\psi$-irreducibility and aperiodicity of $\left\{ Q_{n};n\geq1\right\} $
is immediate. For the drift and minorization conditions of (H\ref{hyp:drift}),
elementary manipulations show that equation \eqref{eq:A_drift} holds
with $V(x)=x^{2}/\left(2\left(1+\delta_{0}\right)\right)+1$ for suitable
$\delta_{0}>0$ and solutions of the minorization condition \eqref{eq:A_minor}
are also easily obtained. Condition (H\ref{hyp:U+_in_V}) is trivially
satisfied because $U$ is non-positive. The density assumption (H\ref{hyp:density})
is satisfied with $\beta_{d}$ proportional to the restriction of Lebesgue measure
to $C_{d}$. Assumption (H\ref{hyp_drift_var}) holds for $\epsilon$
small enough and $\epsilon_{0}=\epsilon/2$.

It is generally not easy to obtain or estimate values for the constants in Theorem \ref{thm:var_bound}. In all the numerical examples which follow, we consider a fixed value of $N$ and consider the relative variance as a function of the $n$ and the initial condition $x$.

The numerical results of Figure \ref{fig:toy_sims} show estimates
of $\mathbb{E}_{x}^{N}\left[\left(\dfrac{\gamma_{n,x}^{N}\left(1\right)}{\gamma_{n,x}(1)}-1\right)^{2}\right]$
with fixed $N=2000$, for various $x$ and $n$, with in each case the expectation
approximated by averaging over $2\times10^{4}$ independent simulations
of the particle system. For this model $\gamma_{n,x}(1)$ can be computed
analytically, and this exact value was used in the estimates. The
linear growth of the relative variance and its dependence on the initial
point $x$ is apparent from the figure.

\begin{figure}
\centering\includegraphics[width=0.95\textwidth]{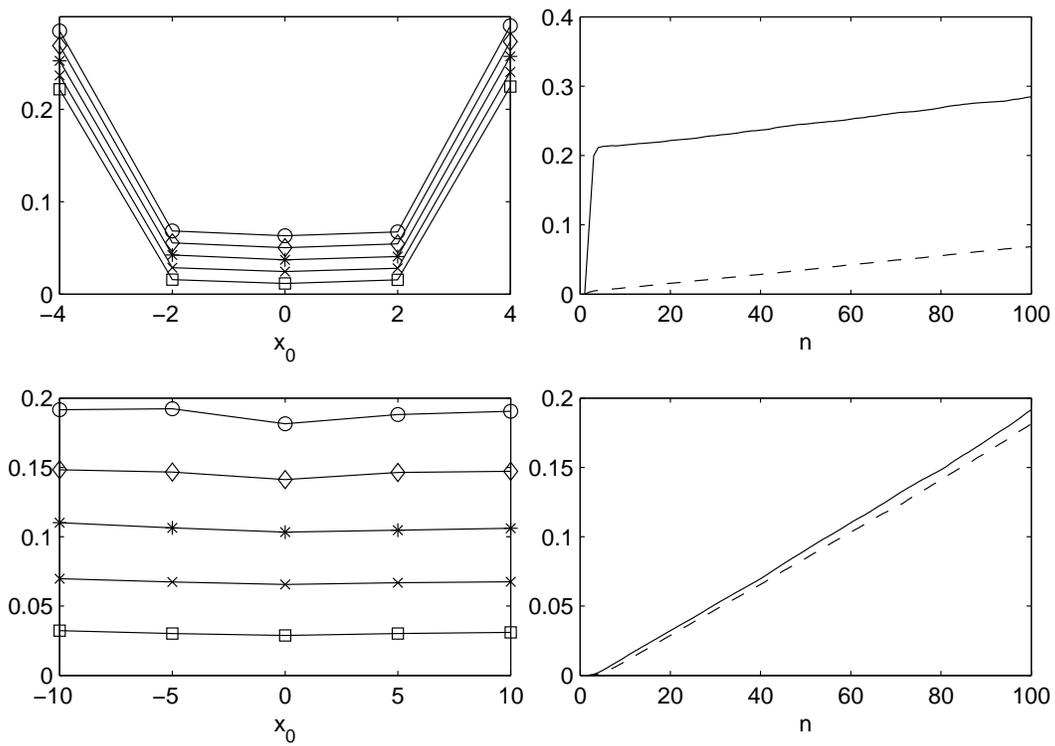} \caption{Top: Gaussian random walk model. Bottom: ergodic autoregression model.
Left: Relative variance vs. initial condition $x_{0}$, at times $\square,n=20$;
$\times,n=40$; $*,n=60$; $\diamond,n=80$; $\circ,n=100$. Right:
Relative variance vs. $n$, from initial conditions (dashed) $x_{0}=0$,
(solid - top) $x_{0}=4$, (solid - bottom) $x_{0}=10$.}

\label{fig:toy_sims} 
\end{figure}

\subsection{Cases with a multiplicative drift assumption on $M$}

The following Lemma shows that condition (H\ref{hyp:drift}) holds
for suitable $U$ when $M$ itself satisfies a multiplicative drift
condition. \begin{lem}\label{lem:M_drift_to_Q_drift} Assume that
there exists $V:\mathsf{X}\rightarrow[1,\infty)$ unbounded, $\delta_{1}>0$,
$d_{1}\geq1$ and for each $d\geq d_{1}$ there exists $b_{d}<\infty$
such that 
\begin{eqnarray}
M\left(e^{V}\right) & \leq & e^{V\left(1-\delta_{1}\right)+b_{d}\mathbb{I}_{C_{d}}},\label{eq:M_mult_drift}
\end{eqnarray}
 and the set $C_{d}=\left\{ x;V(x)\leq d\right\} $ is $(1,\epsilon_d,\nu_d)$-small for $M$, with $\nu_d(C_d)>0$ and $M(C_d)(x)>0$ for all $x$. Then if $U^{+}\in\mathcal{L}_{V}$, $\lim_{r\rightarrow\infty}\left\Vert \mathbb{I}_{C_{r}^{c}}U^{+}\right\Vert _{V}=0$
and for all finite $d$, $\inf_{x\in C_{d}}U(x)>-\infty$,  assumption
(H\ref{hyp:drift}) holds.\end{lem} \begin{proof} Due to the drift
condition (\ref{eq:M_mult_drift}), for any $\delta\in(0,\delta_{1})$,
\begin{eqnarray*}
Q\left(e^{V}\right) & \leq & \exp\left(V\left(1-\delta\right)-\left(\delta_{1}-\delta\right)V+U^{+}+b_{d}\mathbb{I}_{C_{d}}\right),
\end{eqnarray*}
 and due to $\lim_{r\rightarrow\infty}\left\Vert \mathbb{I}_{C_{r}^{c}}U^{+}\right\Vert _{V}=0$,
there exists $\underline{d}$ such that for all $d\geq\underline{d}$,
\[
Q\left(e^{V}\right)\leq\exp\left(V\left(1-\delta\right)+\bar{b}_{d}\mathbb{I}_{C_{d}}\right),
\]
 where $\bar{b}_{d}:=b_{d}+d\left\Vert U^{+}\right\Vert _{V}$, which
verifies the drift part of (H\ref{hyp:drift}). The minorization condition with $m_0=1$ and $Q(C_d)(x)>0$ part are direct as $U(x)$ is bounded below on $C_{d}$. \end{proof}

\subsubsection{Ergodic Autoregression}

Let $\mathsf{X}:=\mathbb{R}$ and $U$ and $M$ be defined by 
\[
U(x):=|x|,\quad M(x,dy):=\dfrac{1}{\sqrt{2\pi}}\exp\left(-\dfrac{\left(y-\alpha x\right)^{2}}{2}\right)dy,
\] for fixed $\left|\alpha\right|<1$. Elementary manipulations then
show that, for $\delta_{0}>0$ and $\underline{d}$ large enough,
$M$ satisfies (\ref{eq:M_mult_drift}) with $V(x)=x^{2}/\left(2\left(1+\delta_{0}\right)\right)+1$.
As per the random walk example, $M$ readily admits minorization on the sublevel sets $C_{d}$.

The potential function
$U$ clearly satisfies (H\ref{hyp:U+_in_V}). Lemma
\ref{lem:M_drift_to_Q_drift} shows that (H\ref{hyp:drift}) is satisfied.
The density assumption (H\ref{hyp:density}) is satisfied for $\beta_{d}$
proportional to Lebesgue measure restricted to $C_{d}$. Again it
is straightforward to check that (H\ref{hyp_drift_var}) is satisfied
for $\epsilon>0$ small enough and $\epsilon_{0}=\epsilon/2$.

Figure \ref{fig:toy_sims} also shows estimates of the relative variance
obtained by simulation for this model with $\alpha=0.4$ and using
$N=10^{4}$ particles, averaged over $10^{4}$ independent realizations.
Again the linear growth of the variance is apparent, but there appears
to be less variation with respect to the initial condition than in
the random walk example.

\subsubsection{Cox-Ingersoll-Ross Process}\label{sec:CIR}

The Cox-Ingersoll-Ross (CIR) process, \citep{cox1985theory}, is a
diffusion process that is typically used in financial applications
to capture mean-reverting behaviour and state-dependent volatility, which is thought to occur in many real scenarios. The process is defined via the stochastic
differential equation: 
\[
dX_{t}=\theta\left(\mu-X_{t}\right)dt+\sigma\sqrt{X_{t}}dW_{t}
\]
where $\{W_{t}\}$ is standard Brownian motion, $\theta>0$
is the mean-reversion rate, $\mu>0$ is the level of mean-reversion
and $\sigma>0$ is the volatility. We assume that $\dfrac{2\theta\mu}{\sigma^{2}}>1$
so that the process is stationary and never touches zero.

Throughout the remainder of section \ref{sec:CIR}, for $\Delta>0$ we denote by $M^{\Delta}$ the transition probability from any time $t$ to $t+\Delta$ of the CIR process with parameters $\theta,\mu,\sigma$. The following lemma identifies a drift function for $M^{\Delta}$, exhibiting a trade-off between growth rate
of the drift function specified by a parameter $s$, the parameters of the CIR process and the time step size $\Delta$.

\begin{lem} For $s>0$ and $\Delta>0$, consider the candidate drift function $V:\mathbb{R}_{+}\rightarrow[1,\infty)$,
defined by 
\begin{equation}
V(x):=1+\frac{4\theta sx}{\sigma^{2}\left(1-e^{-\theta\Delta}\right)}.\label{eq:CIR_drift}
\end{equation}
Then subject to the conditions: 
\begin{equation}
s\in\left(0,\frac{1-e^{-\theta\Delta}}{2}\right),\quad\delta\in\left(0,1-\frac{e^{-\theta\Delta}}{1-2s}\right),\quad d\geq\frac{1-2\theta\mu\log\left(1-2s\right)/\sigma^{2}}{1-e^{-\theta\Delta}/\left(1-2s\right)-\delta}=:\underline{d},\label{eq:CIR_drift_params}
\end{equation}
the following multiplicative drift condition is satisfied: 
\begin{eqnarray*}
M^\Delta\left(e^{V}\right) & \leq & e^{V(1-\delta)+b_{d}\mathbb{I}_{C_{d}}},
\end{eqnarray*}
with $V$ as in (\ref{eq:CIR_drift}) and $b_{d}:=\dfrac{de^{-\theta\Delta}}{1-2s}-\dfrac{2\theta\mu}{\sigma^{2}}\log\left(1-2s\right)+1$.\label{lem:cir}\end{lem}
\begin{proof} For $t\geq 0$ define 
\[
c_{t}:=\frac{2\theta}{\sigma^{2}\left(1-e^{-\theta t}\right)},\quad\kappa:=\frac{4\theta\mu}{\sigma^{2}},
\]
and the scaled process $Z_{t}:=2c_{t}X_{t}$. Conditional on $X_{0}=x$,
$Z_{t}$ has a non-central chi-square distribution with degree of
freedom $\kappa$ and non-centrality parameter taking the value $2c_{t}xe^{-\theta t}$ \citep{cox1985theory}.
We then have for any $x\in\mathsf{X}$,
\begin{eqnarray*}
M^{\Delta}\left(e^{V}\right)(x) & = & \mathbb{E}_{x}\left[\exp\left(sZ_{\Delta}\right)\right]\exp(1)\\
 & = & \exp\left[2c_{\Delta}xs\left(\frac{e^{-\theta\Delta}}{1-2s}\right)-\frac{\kappa}{2}\log\left(1-2s\right)+1\right]\\
 & \leq & \exp\left[V(x)\left(\frac{e^{-\theta\Delta}}{1-2s}\right)-\frac{\kappa}{2}\log\left(1-2s\right)+1\right].
\end{eqnarray*}
where the equalities hold due to the existence of the moment generating
function $\mathbb{E}_{x}\left[\exp\left(sZ_{t}\right)\right]$, for
$s<1/2$, which is satisfied under the conditions of (\ref{eq:CIR_drift_params}).
Under these conditions we also then have for $d\geq\underline{d}$ and $x\notin C_{d}$, 
\begin{eqnarray*}
M^\Delta\left(e^{V}\right)(x) & \leq & \exp\left[V(x)\left(1-\delta\right)-d\left(1-\frac{e^{-\theta\Delta}}{1-2s}-\delta\right)-\frac{\kappa}{2}\log\left(1-2s\right)+1\right]\\
 & \leq & \exp\left[V(x)\left(1-\delta\right)\right],
\end{eqnarray*}
and for $x \in C_{d}$,
\begin{eqnarray*}
M\left(e^{V}\right)(x) & \leq & \exp\left[d\left(\frac{e^{-\theta\Delta}}{1-2s}\right)-\frac{\kappa}{2}\log\left(1-2s\right)+1\right]=\exp\left(b_{d}\right).
\end{eqnarray*}
\end{proof}
We will consider as an example the case where the Markov chain $\{X_n\}$ is the skeleton of the CIR process over a discrete time grid of spacing $\Delta$ and $U(x):=\alpha\log x$
for some fixed $\alpha$. Lemmata \ref{lem:M_drift_to_Q_drift}
and \ref{lem:cir} establish that (H\ref{hyp:drift})-(H\ref{hyp:U+_in_V}) are satisfied
and one can check (H\ref{hyp:density})-(H\ref{hyp_drift_var}) are
satisfied similarly to the previous example.

Figure \ref{fig:cir} displays estimates of the relative variance
for this model, computed via simulation, when $\Delta=0.01$, (i.e. $M \equiv M^{0.01}$), $\alpha=0.01$,  $\theta=10$, $\mu=1$, and $\sigma=0.1$.
This was obtained using $N=10^{3}$ particles, averaged over $3\times10^{3}$
independent realizations. Again the linear growth of the relative
variance is present for different initial conditions. Note one may
interpret $\gamma_{100,x}(1)$ as the
geometric mean $\mathbb{E}_{x}[\prod_{k=0}^{99}X_{k}^{1/100}]$, which can be used for prediction in a variety of financial applications. 

\begin{figure}
\centering \includegraphics[width=0.7\textwidth]{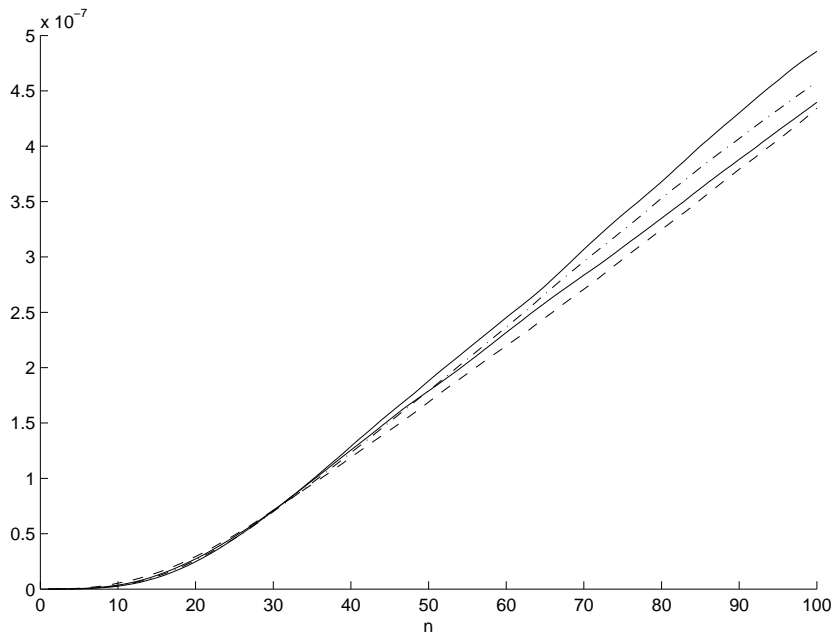}
\caption{Cox-Ingersoll-Ross process. Relative variance vs. $n$, from initial conditions $x_{0}=0.1$  (dashed),
$x_{0}=1$ (solid - bottom), $x_{0}=3$ (dot-dashed), $x_{0}=10$ (solid top).}
\label{fig:cir} 
\end{figure}

\section{Summary}

In this paper we have established a linear-in-$n$ bound on
the non-asymptotic variance associated with particle approximations
of time-homogeneous Feynman-Kac formulae, under assumptions that can
be verified on non-compact state-spaces. 

There are several possible extensions to this work. Firstly, to consider non-homogeneous
Feynman-Kac formulae, which occur routinely in applications such
as filtering and Bayesian statistics. Secondly, an important developing area in the analysis of sequential
Monte Carlo methods is the case when the dimension of the state-space
can be very large \citep{beskos}. Such analysis has relied on classical
geometric drift conditions and it would be interesting to consider the role of multiplicative drift conditions in this context.

\subsubsection*{Acknowledgements}

We would like to thank the associate editor and the referee for some very useful comments that have lead to considerable improvements in the paper.
The first and third authors acknowledge the assistance of the London
Mathematical Society for their funding, via a research in pairs grant.
The second author was supported by the EPSRC programme grant on Control
For Energy and Sustainability EP/G066477/1.

\appendix

\section{Proofs and Auxiliary Results for Section \ref{sec:Multiplicative-ergodicity}}

\label{app:MET_proofs}

\begin{proof}{[}Proof of Lemma \ref{lem:spr_bounded_below}{]} Fix
any $d\geq\underline{d}$. The upper bound of (\ref{eq:FK_well-defined})
is an immediate consequence of the inequality $Q\left(e^{V}\right)/e^{V}\leq e^{b_{d}}$,
implied by (\ref{eq:A_drift}).

For the upper bound of (\ref{eq:spr_lowerbound}), use the standard
inequality $\xi_{v}\left(Q\right)\leq\interleave Q\interleave_{v}$
and then also due to the drift condition in (\ref{eq:A_drift}), $\interleave Q\interleave_{v}<\infty$.
Now consider the lower bound. It is claimed that for any $k\geq3$
and $1\leq j\leq k-1$, 
\begin{eqnarray}
Q_{km_{0}}\left(e^{V}\right)(x) & \geq & Q_{\left(k-j\right)m_{0}}\left(\mathbb{I}_{C_{d}}\right)(x)\epsilon_{d}^{j}\nu_{d}\left(C_{d}\right)^{j-1}\nu_{d}\left(e^{V}\right),\quad\forall x\in\mathsf{X},\label{eq:Q_n_recurse}
\end{eqnarray}
 where $m_{0}$ is as in (H\ref{hyp:drift}). For each $k$, the claim
is verified by induction in $j$; fix $k\geq3$ arbitrarily. For $j=1$,
\[
Q_{km_{0}}\left(e^{V}\right)(x)\geq Q_{\left(k-1\right)m_{0}}\left(\mathbb{I}_{C_{d}}Q_{m_0}\left(e^{V}\right)\right)(x)\geq Q_{\left(k-1\right)m_{0}}\left(\mathbb{I}_{C_{d}}\right)(x)\epsilon_{d}\nu_{d}\left(e^{V}\right)
\]
 which initializes the induction. Now assume that (\ref{eq:Q_n_recurse})
holds at rank $1\leq j<k-1$. Then at rank $j+1$, applying the induction
hypothesis 
\begin{eqnarray*}
Q_{km_{0}}\left(e^{V}\right)(x) & \geq & Q_{\left(k-j-1\right)m_{0}}\left(\mathbb{I}_{C_{d}}Q_{m_0}\left(\mathbb{I}_{C_{d}}\right)\right)(x)\epsilon_{d}^{j}\nu_{d}\left(C_{d}\right)^{j-1}\nu_{d}\left(e^{V}\right)\\
 & \geq & Q_{\left(k-j-1\right)m_{0}}\left(\mathbb{I}_{C_{d}}\right)(x)\epsilon_{d}^{j+1}\nu_{d}\left(C_{d}\right)^{j}\nu_{d}\left(e^{V}\right),\quad\forall x\in\mathsf{X},
\end{eqnarray*}
 where \eqref{eq:A_minor} has been applied, thus the claim is verified.

Now applying (\ref{eq:Q_n_recurse}) with $j=k-1$ gives, 
\begin{eqnarray*}
\frac{Q_{km_{0}}\left(e^{V}\right)(x)}{e^{V(x)}} & \geq & \frac{Q_{m_{0}}\left(\mathbb{I}_{C_{d}}\right)(x)}{e^{V(x)}}\epsilon_{d}^{k-1}\nu_{d}\left(C_{d}\right)^{k-2}\nu_{d}\left(e^{V}\right)>0,\quad\forall x\in\mathsf{X},
\end{eqnarray*}
 which implies that 
\[
\interleave Q_{km_{0}}\interleave_{v}^{1/\left(km_{0}\right)}\geq\epsilon_{d}^{1-1/\left(km_{0}\right)}\nu_{d}\left(C_{d}\right)^{1-2/\left(km_{0}\right)}\nu_{d}\left(e^{V}\right)^{1/\left(km_{0}\right)}\left[\sup_{x\in\mathsf{X}}\frac{Q\left(\mathbb{I}_{C_{d}}\right)(x)}{e^{V(x)}}\right]^{1/\left(km_{0}\right)}.
\]
 Taking $k\rightarrow\infty$ is enough to verify (\ref{eq:spr_lowerbound}),
as $\lim_{n\rightarrow\infty}\interleave Q^{n}\interleave_{v}^{1/n}$
always exists by subadditivity. \end{proof}

\begin{proof}{[}Proof of Lemma \ref{lem:Q_sep}{]} Set $r\geq\underline{d}$
arbitrarily and let $\widehat{Q}^{(r)}:=\mathbb{I}_{C_{r}}Q$. For $n\geq1$, denote by $\widehat{Q}^{(r)}_n$ the $n$-fold iterate of $\widehat{Q}^{(r)}$. 

Then
under (H\ref{hyp:U+_in_V}), 
\begin{eqnarray*}
\widehat{Q}^{(r)}_{t_{0}+1}(x,A) & = & \mathbb{E}_{x}\left[\prod_{n=0}^{t_{0}}\mathbb{I}_{C_{r}}\left(X_{n}\right)\exp\left(U\left(X_{n}\right)\right)\mathbb{I}_{A}\left(X_{t_{0}+1}\right)\right]\\
 & \leq & \exp\left(rt_{0}\left\Vert U^{+}\right\Vert _{V}\right)\mathbb{E}_{x}\left[\prod_{n=0}^{t_{0}}\mathbb{I}_{C_{r}}\left(X_{n}\right)\exp\left(U\left(X_{t_{0}}\right)\right)\mathbb{I}_{A}\left(X_{t_{0}+1}\right)\right],\quad\forall x\in\mathsf{X},A\in\mathcal{B}\mathsf{\left(X\right)},
\end{eqnarray*}
 and therefore under (H\ref{hyp:density}),
\begin{equation}
\widehat{Q}^{(r)}_{t_{0}+1}(x,A)\leq\beta_{r}^{*}(A):=\exp\left(rt_{0}\left\Vert U^{+}\right\Vert _{V}\right)\int_{C_{r}}\beta_{r}(dy)Q(y,A),\quad\forall x\in\mathsf{X},A\in\mathcal{B}\mathsf{\left(X\right)}.\label{eq:Q_r_maj}
\end{equation}
 Lemma B3 of \citep{mc:the:KM05} then implies that $\widehat{Q}^{(r)}_{2t_{0}+2}$
is $v$-separable. 

In order to establish that $Q_{2t_{0}+2}$ is $v$-separable, we will prove that $\interleave Q_{2t_0+2}-\widehat{Q}^{(r)}_{2t_0+2}\interleave_v$ can be made arbitrarily small through suitable choice of $r$. By decomposing the difference $Q_{2t_{0}+2}-\widehat{Q}^{(r)}_{2t_{0}+2}$ in a telescoping fashion and applying the sub-additive and sub-multiplicative properties of the operator norm we obtain:
\begin{align}
\interleave Q_{2t_{0}+2}-\widehat{Q}^{(r)}_{2t_{0}+2}\interleave_v \leq&\sum_{n=0}^{2t_{0}+1}\interleave \widehat{Q}^{(r)}_{2t_{0}+2-(n+1)}Q_{n+1}-\widehat{Q}^{(r)}_{2t_0+2-n}Q_{n}\interleave_v,\nonumber\\
\leq&\interleave Q-\widehat{Q}^{(r)}\interleave_v\sum_{n=0}^{2t_{0}+1}\interleave\widehat{Q}^{(r)}_{2t_{0}+2-(n+1)}\interleave_v  \interleave Q_{n}\interleave_v.\label{eq:Q_telescope}
\end{align}
Now for any $n\geq 0$, $\sup_r\interleave\widehat{Q}^{(r)}_{n}\interleave_v\leq \interleave Q_n \interleave_v<\infty$, where the final inequality follows from equation \eqref{eq:FK_well-defined} of Lemma \ref{lem:spr_bounded_below}, and by
\eqref{eq:approx_Q} we have $\interleave Q-\widehat{Q}^{(r)}\interleave_{v}\rightarrow0$
as $r\rightarrow\infty$. Therefore it follows from \eqref{eq:Q_telescope} that $\interleave Q_{2t_{0}+2}-\widehat{Q}_{2t_{0}+2}^{(r)}\interleave_{v}\rightarrow0$
as $r\rightarrow\infty$,  so we conclude that $Q_{2t_{0}+2}$ is $v$-separable. This completes the proof.
\end{proof}

The following lemma considers the twisted kernel $\check{P}$ defined
in \eqref{eq:twisted_kernel_defn}.

\begin{lem} \label{lem:P_check_M-drift}Assume (H\ref{hyp:irreducible})-(H\ref{hyp:density}).
Then there exists $\delta_{0}\in(0,\delta)$, $d_{0}\geq1$ and for
any $d\geq d_{0},$ there exists $\check{b}_{d}<\infty$ such that
\begin{eqnarray}
\check{P}\left(e^{\check{V}}\right) & \leq & e^{\check{V}-\delta_{0}V+\check{b}_{d}\mathbb{I}_{C_{d}}},\label{eq:P_check_m_drift}\\
\sup_{x\in C_{d}}e^{\check{V}(x)} & < & \infty,\label{eq:V_check _bounded_on_C}
\end{eqnarray}
 where $\check{V}:\mathsf{X}\rightarrow[1,\infty)$ is defined by
$\check{V}(x):=V(x)-\log h_{0}(x)+\log\left\Vert h_{0}\right\Vert _{v}$.
Furthermore, there exists $\rho<1$, depending only on $d_{0}$ and
$\delta_{0}$, and for any $d\geq d_{0}$ there exists $\check{b}_{d}'<\infty$
such that 
\begin{eqnarray}
\check{P}\left(e^{\check{V}}\right) & \leq & \rho e^{\check{V}}+\check{b}_{d}'\mathbb{I}_{C_{d}}.\label{eq:twisted_geo_drift}
\end{eqnarray}
 \end{lem}

\begin{proof} Under the assumptions of the lemma, we have already
seen via \citep[Proposition 2.8]{mc:the:KM05} that the twisted kernel
is well defined. First consider, \eqref{eq:P_check_m_drift}; under
(H\ref{hyp:drift}), setting $\delta_{0}\in(0,\delta)$, for any $d\geq\underline{d}$,
\begin{eqnarray*}
\check{P}\left(\frac{e^{V}}{h_{0}}\right) & = & \lambda^{-1}h_{0}^{-1}Q\left(e^{V}\right)\\
 & \leq & \exp\left(V-\log h_{0}-\delta_{0}V-(\delta-\delta_{0})V-\log\lambda+b_{d}\mathbb{I}_{C_{d}}\right).
\end{eqnarray*}
 As $V$ is unbounded, there exists $d_{0}$ such that for all $d\geq d_{0}$,
equation (\ref{eq:P_check_m_drift}) holds with $\check{b}_{d}:=b_{d}-\log\lambda$.

For (\ref{eq:V_check _bounded_on_C}) by iteration of the eigenfunction
equation, we have that for any $d\geq d_{0}$, 
\begin{align*}
h_{0}(x) & =\lambda^{-m_{0}}Q_{m_{0}}\left(h_{0}\right)(x)\geq\epsilon_{d}\nu_{d}\left(h_{0}\right),\quad\forall x\in C_{d}
\end{align*}
 where we apply the minorization part of (H\ref{hyp:drift}) to obtain
the inequality.

It remains to establish \eqref{eq:twisted_geo_drift}. First considering the case $x\notin C_{d}$,
 \eqref{eq:P_check_m_drift} implies that $\check{P}\left(e^{\check{V}}\right)(x)\leq e^{\check{V}(x)-\delta_{0}V(x)}\leq e^{\check{V}(x)-\delta_{0}d}$
so that \eqref{eq:twisted_geo_drift} holds with $\rho:=e^{-\delta_{0}d_{0}}$.
For $x\in C_d$, equation (\ref{eq:P_check_m_drift}) shows that \eqref{eq:twisted_geo_drift}  
with $\check{b}_{d}':=\exp(d-\log\epsilon_{d}-\log\nu_{d}\left(h_{0}\right)$
$+\check{b}_{d}+\log\left\Vert h_{0}\right\Vert _{v})$. \end{proof}

\section{Proofs and Auxiliary Results for Section \ref{sec:Non-asymptotic-variance}}
\label{app:variance}

In this appendix we detail the proofs and auxiliary results that are used in Section \ref{sec:Non-asymptotic-variance}. The proofs and results are provided in a logical order; that is, each result at most depends on the preceding one(s). In particular, the proof of Lemma \ref{lem:E_check_prod_bounds} follows the proof of Lemma \ref{lem:M_drift_var}.

\begin{lem} \label{lem:M_drift_var}Assume (H\ref{hyp:irreducible})-(H\ref{hyp_drift_var}).Then
there exists $\bar{\rho}<1$, $d_{0}\geq1$ and for any $d\geq d_{0}$
there exists $\bar{b}_{d}<\infty$ and $\bar{b}_{d}'<\infty$ such that 
\begin{eqnarray}
\check{P}\left(e^{V^{*}}\right) & \leq & e^{V^{*}-V+\bar{b}_{d}\mathbb{I}_{C_{d}}}\label{eq:drift_star}\\
\check{P}\left(e^{V^{*}}\right) & \leq & \bar{\rho}e^{V^{*}}+\bar{b}_{d}'\mathbb{I}_{C_{d}},\label{eq:drift_geo_star}
\end{eqnarray}
 where $V^{*}$ is as in equation (\ref{eq:V_star_defn}).\end{lem}
\begin{proof} 
Under the assumptions of the lemma, Theorem \ref{thm:MET} holds,
the eigenfunction $h_{0}\in\mathcal{L}_{v}$, and the twisted kernel
is well defined. Then under (H\ref{hyp_drift_var}), we have for any
$d\geq\underline{d}$, 
\begin{eqnarray*}
\check{P}\left(\frac{e^{V\left(1+\epsilon\right)}}{h_{0}}\right) & = & \lambda^{-1}h_{0}^{-1}Q\left(e^{V\left(1+\epsilon\right)}\right)\\
 & \leq & \exp\left(V\left(1+\epsilon\right)-\log h_{0}-V-\epsilon_{0}V-\log\lambda+b_{d}^{*}\mathbb{I}_{C_{d}}\right).
\end{eqnarray*}
 As $V$ is unbounded, there exists $d_{0}$ such that for all $d\geq d_{0}$,
equation (\ref{eq:drift_star}) holds with $\bar{b}_{d}:=b_{d}^{*}-\log\lambda$.
The proof of (\ref{eq:drift_geo_star}) then follows exactly as in
the proof of Lemma \ref{lem:P_check_M-drift}. \end{proof}

\begin{proof}{[}Proof of Lemma \ref{lem:E_check_prod_bounds}{]}
We first consider some bounds on iterates of the twisted kernel. Standard
iteration of the geometric drift condition in equation (\ref{eq:drift_geo_star})
shows that there exists a finite constant $c_{1}$ such that 
\begin{equation}
\sup_{n\geq 0}\check{P}_{n}\left(v^{*}\right)(x)\leq c_{1}v^{*}(x),\quad x\in\mathsf{X},\label{eq:v_star_iterate_bound}
\end{equation}
 and then due to the multiplicative drift condition in equation (\ref{eq:drift_star}),
\begin{equation}
\sup_{n>0}v(x)\check{P}_{n}\left(v^{*}\right)(x)=\sup_{n\geq0}v(x)\check{P}\check{P}_{n-1}\left(v^{*}\right)(x)\leq c_{1}v(x)\check{P}\left(v^{*}\right)(x)\leq cv^{*}(x),\quad x\in\mathsf{X},\label{eq:v_star_iterate_bound_m}
\end{equation}
 where $c:=c_{1}e^{\bar{b}_{d}}$.

In order to prove (\ref{eq:E_check_bounded_cs}) first fix arbitrarily
$n\geq1$, $1\leq s\leq n+1$ and $\left(i_{1},\ldots,i_{s}\right)\in\mathcal{I}_{n,s}$.
The proof is via a backward inductive argument through the coalescent
time indices. Assume that at rank $1<j<s$, 
\begin{eqnarray}
v(x)\check{\mathbb{E}}_{x}\left[\prod_{k\in\left\{ i_{j+1}-i_{j},\ldots,i_{s}-i_{j}\right\} }v\left(\check{X}_{k}\right)v^{*}\left(\check{X}_{n+1-i_{j}}\right)\right] & \leq & c^{s+1-j}v^{*}(x).\label{eq:E_check_bound_lem_ind_hyp}
\end{eqnarray}
 Assuming (\ref{eq:E_check_bound_lem_ind_hyp}) is true, then at rank
$j-1$, 
\begin{eqnarray*}
 &  & v(x)\check{\mathbb{E}}_{x}\left[\prod_{k\in\left\{ i_{j}-i_{j-1},\ldots,i_{s}-i_{j-1}\right\} }v\left(\check{X}_{k}\right)v^{*}\left(\check{X}_{n+1-i_{j-1}}\right)\right]\\
 &  & =v(x)\int\check{P}_{i_{j}-i_{j-1}}\left(x,dx'\right)v(x')\check{\mathbb{E}}_{x'}\left[\prod_{k\in\left\{ i_{j+1}-i_{j},\ldots,i_{s}-i_{j}\right\} }v\left(\check{X}_{k}\right)v^{*}\left(\check{X}_{n+1-i_{j}}\right)\right]\\
 &  & \leq c^{s+1-j}v(x)\int\check{P}_{i_{j}-i_{j-1}}\left(x,dx'\right)v^{*}(x')\\
 &  & \leq c^{s+1-\left(j-1\right)}v^{*}(x),
\end{eqnarray*}
 where the final inequality is due to equation (\ref{eq:v_star_iterate_bound_m}).
Furthermore 
\begin{eqnarray*}
v(x)\check{\mathbb{E}}_{x}\left[v^{*}\left(\check{X}_{n+1-i_{s}}\right)\right] & = & v(x)\check{P}_{n+1-i_{s}}\left(v^{*}\right)(x)\leq cv^{*}(x),
\end{eqnarray*}
 where the inequality is again due to (\ref{eq:v_star_iterate_bound_m})
and therefore at rank $j=s-1$, 
\begin{eqnarray*}
v(x)\check{\mathbb{E}}_{x}\left[\prod_{k=\left(i_{s}-i_{s-1}\right)}v\left(\check{X}_{k}\right)v^{*}\left(\check{X}_{n+1-i_{s-1}}\right)\right] & = & v(x)\int\check{P}_{i_{s}-i_{s-1}}\left(x,dx'\right)v(x')\check{\mathbb{E}}_{x'}\left[v^{*}\left(\check{X}_{n-i_{s}}\right)\right]\\
 & \leq & cv(x)\int\check{P}_{i_{s}-i_{s-1}}\left(x,dx'\right)v^{*}(x')\\
 & \leq & c^{2}v^{*}(x').
\end{eqnarray*}
 The above arguments prove that (\ref{eq:E_check_bound_lem_ind_hyp})
holds at rank $j=1$ and the proof of the Lemma is then also complete
as $n+1$, $1\leq s\leq n+1$ and $\left(i_{1},\ldots,i_{s}\right)\in\mathcal{I}_{n,s}$
were arbitrary. \end{proof}

\begin{lem} \label{lem:lam_D_Q_bound}Assume (H\ref{hyp:irreducible})-(H\ref{hyp_drift_var}).
Then there exists $c<\infty$ depending only on the quantities in (H\ref{hyp:irreducible})-(H\ref{hyp_drift_var})
such that for any $n\geq1$ and $\varphi:\mathsf{X}\rightarrow\mathbb{R}_{0}^{+}$,
\begin{equation}
\lambda^{-2n}DQ_{n}^{\otimes2}\left(\varphi\otimes v\right)(x,x')\leq cv(x)h_{0}(x)\check{P}_{n}\left(\frac{\varphi}{h_{0}}\right)(x),\quad\left(x,x'\right)\in\mathsf{X},\label{eq:QD_recurse}
\end{equation}
 where $v$ is as in (H\ref{hyp:drift}), and $\lambda$ and $h_{0}\in\mathcal{L}_{v}$
are respectively the eigenvalue and eigenfunction as in Theorem \ref{thm:MET}.\end{lem}
\begin{proof} By standard iteration of the geometric drift condition
in equation \eqref{eq:twisted_geo_drift} of Lemma \ref{lem:P_check_M-drift},
there is a finite constant $c$ such that 
\begin{equation}
\sup_{n\geq 0}\check{P}_{n}\left(\check{v}\right)(x)\leq c\check{v}(x),\quad x\in\mathsf{X}.\label{eq:twisted_drift_iterate}
\end{equation}
 Then due to the definition of the twisted kernel and $\check{v}$
(see Lemma \ref{lem:P_check_M-drift}), there exists a constant $c$
such that for any $n\geq 1$, and $\varphi:\mathsf{X}\rightarrow\mathbb{R}_{0}^{+}$,
\begin{eqnarray}
\lambda^{-2n}Q_{n}^{\otimes2}\left(\varphi\otimes v\right)(x,x') & = & h_{0}(x)h_{0}(x')\check{P}_{n}^{\otimes2}\left(\frac{\varphi}{h_{0}}\otimes\frac{v}{h_{0}}\right)\left(x,x'\right)\nonumber \\
 & \leq & ch_{0}(x)h_{0}(x')\check{P}_{n}^{\otimes2}\left(\frac{\varphi}{h_{0}}\otimes\check{v}\right)\left(x,x'\right)\nonumber \\
 & \leq & ch_{0}(x)\check{P}_{n}\left(\frac{\varphi}{h_{0}}\right)\left(x\right)v(x'),\quad(x,x')\in\mathsf{X}^{2},\label{eq:Q_recurse}
\end{eqnarray}
 where the final inequality is due to \eqref{eq:twisted_drift_iterate}.
\end{proof} \begin{lem} \label{lem:DQDQ_bound}Assume (H\ref{hyp:irreducible})-(H\ref{hyp_drift_var}).
Then there exists $c<\infty$ depending only on the quantities in (H\ref{hyp:irreducible})-(H\ref{hyp_drift_var})
such that for any $m\geq1$, $n\geq0$ and $\left(x,x'\right)\in\mathsf{X}^{2}$,
\begin{eqnarray*}
\lambda^{-2\left(m+n\right)}DQ_{m}^{\otimes2}DQ_{n}^{\otimes2}\left(v^{1/2}\otimes v^{1/2}\right)(x,x') & \leq & cv(x)h_{0}(x)\check{\mathbb{E}}_{x}\left[v^{*}\left(\check{X}_{m}\right)\right].
\end{eqnarray*}
 \end{lem} \begin{proof} Throughout the proof $c$ is a finite constant
whose value may change on each appearance. 

When $n=0$,  
\begin{eqnarray*}
\lambda^{-2\left(m+n\right)}DQ_{m}^{\otimes2}DQ_{n}^{\otimes2}\left(v^{1/2}\otimes v^{1/2}\right)(x,x^{\prime}) & = & \lambda^{-2\left(m+n\right)}Q_{m}^{\otimes2}\left(D\left(v^{1/2}\otimes v^{1/2}\right)\right)(x,x)\\
 & = & \lambda^{-2\left(n+m\right)}Q_{m}^{\otimes2}\left(v\otimes1\right)(x,x)\\
 & \leq & cv(x)h_{0}(x)\check{P}_{m}\left(\frac{v}{h_{0}}\right)(x)\\
 & \leq & cv(x)h_{0}(x)\check{P}_{m}\left(v^{*}\right)(x)\\
 & = & cv(x)h_{0}(x)\check{\mathbb{E}}_{x}\left[v^{*}\left(\check{X}_{m}\right)\right],
\end{eqnarray*}
where the first inequality is due to Lemma \ref{lem:lam_D_Q_bound} and the second inequality is due to the definition of $v^*$.

Now consider the case
$n\geq1$. We have 
\begin{eqnarray*}
\lambda^{-2n}DQ_{n}^{\otimes2}\left(v^{1/2}\otimes v^{1/2}\right)(x,x') & \leq & cv(x)h_{0}(x)\check{P}_{n}\left(\frac{v^{\left(1+\epsilon_{0}\right)}}{h_{0}}\right)(x)\\
 & \leq & cv(x)h_{0}(x)\check{P}_{n}\left(v^{*}\right)(x)\\
 & \leq & cv(x)h_{0}(x)\check{\mathbb{E}}_{x}\left[v^{*}\left(\check{X}_{n}\right)\right]v(x'),
\end{eqnarray*}
where we have used $v\geq1$, Lemma \ref{lem:lam_D_Q_bound} with $\varphi=v$, the definition of $v^*$ and again $v\geq1$. A further application of Lemma \ref{lem:lam_D_Q_bound} with $\varphi(x)=v(x)h_{0}(x)\check{\mathbb{E}}_{x}\left[v^{*}\left(\check{X}_{n}\right)\right]$
and an application of Lemma \ref{lem:E_check_prod_bounds} yields:
\begin{eqnarray*}
\lambda^{-2\left(m+n\right)}DQ_{m}^{\otimes2}DQ_{n}^{\otimes2}\left(v^{1/2}\otimes v^{1/2}\right)(x,x) & \leq & c^{2}v(x)h_{0}(x)\check{\mathbb{E}}_{x}\left[\check{\mathbb{E}}_{\check{X}_{m}}\left[v\left(\check{X}_{0}\right)v^{*}\left(\check{X}_{n}\right)\right]\right]\\
 & \leq & c^{2}v(x)h_{0}(x)\check{\mathbb{E}}_{x}\left[v^{*}\left(\check{X}_{m}\right)\right].
\end{eqnarray*}
This completes the proof. \end{proof}


\begin{proof}{[}Proof of Proposition \ref{prop:Gamma_bound}{]} The
starting point of the proof is to write, using the definition of the
twisted kernel, 
\[
\bar{\Gamma}_{n,x}^{(i_{1},...i_{s})}\left(F\right)=\frac{\lambda^{-2n}\Gamma_{n,x}^{\left(i_{1},...,i_{s}\right)}\left(F\right)}{\lambda^{-2n}\gamma_{n,x}\left(1\right)^{2}}=\frac{\lambda^{-2n}\Gamma_{n,x}^{\left(i_{1},...,i_{s}\right)}\left(F\right)}{h_{0}^{2}(x)\check{\mathbb{E}}_{x}\left[1/h_{0}\left(\check{X}_{n}\right)\right]^{2}}.
\]
 Thus in order prove (\ref{eq:twisted_E_bound}), we need to prove
\begin{equation}
\lambda^{-2n}h_{0}^{-2}(x)\Gamma_{n,x}^{\left(i_{1},...,i_{s}\right)}\left(F\right)\leq\left\Vert F\right\Vert _{v^{1/2},2}c^{s+1}\frac{v(x)}{h_{0}(x)}\check{\mathbb{E}}_{x}\left[\prod_{k\in\left\{ i_{1},\ldots,i_{s-1}\right\} }v\left(\check{X}_{k}\right)v^{*}\left(\check{X}_{i_{s}}\right)\right],\label{eq:numerator_bound}
\end{equation}
 for each $n\geq1$, $0\leq s\leq n+1$ and each possible configuration
of the coalescent time indices $(i_{1},...,i_{s})\in\mathcal{I}_{n,s}$. We will consider first the case $s>1$ and then $s\leq 1$. Throughout the remainder of the proof, $c$ denotes a finite and positive
constant, whose value may change on each appearance but depends only on
the constants in (H\ref{hyp:irreducible})-(H\ref{hyp_drift_var}).

Consider the case $s>1$. It is claimed that there exists a
finite constant $c$ such that for any $n\geq1$, $\left(x,x'\right)\in\mathsf{X}^{2}$,
$F\in\mathcal{L}_{v^{1/2},2}$, $1<s\leq n+1,$ and any $\left(i_{1},\ldots,i_{s}\right)\in\mathcal{I}_{n,s}$,
\begin{eqnarray}
 &  & \lambda^{-2\left(n-i_{1}\right)}DQ_{i_{2}-i_{1}}^{\otimes2}\ldots DQ_{i_{s}-i_{s-1}}^{\otimes2}DQ_{n-i_{s}}^{\otimes2}\left(v^{1/2}\otimes v^{1/2}\right)(x,x')\nonumber \\
 &  & \leq c^{s+1}v(x)h_{0}(x)\check{\mathbb{E}}_{x}\left[\prod_{k\in\left\{ i_{2}-i_{1},\ldots,i_{s-1}-i_{1}\right\} }v\left(\check{X}_{k}\right)v^{*}\left(\check{X}_{i_{s}-i_{1}}\right)\right],\label{eq:claim_induction}
\end{eqnarray}
with the convention that the product is equal to unity when $s=2$.
For a given $n$, the claim is proved by backward induction through
the coalescent time indices. The inductive hypothesis is that at rank
$1\leq j\leq s-1$, 
\begin{eqnarray}
 &  & \lambda^{-2\left(n-i_{j}\right)}DQ_{i_{j+1}-i_{j}}^{\otimes2}\ldots DQ_{i_{s}-i_{s-1}}^{\otimes2}DQ_{n-i_{s}}^{\otimes2}\left(v^{1/2}\otimes v^{1/2}\right)(x,x')\nonumber\\
 &  & \leq c^{s-j+1}v(x)h_{0}(x)\check{\mathbb{E}}_{x}\left[\prod_{k\in\left\{ i_{j+1}-i_{j},\ldots,i_{s-1}-i_{j}\right\} }v\left(\check{X}_{k}\right)v^{*}\left(\check{X}_{i_{s}-i_{j}}\right)\right],\label{eq:rank_j} 
\end{eqnarray}
with the convention that the product equals unity when $j+1=s$.
 
To initialise the induction, we have at rank $j=s-1$ that the left hand side of \eqref{eq:rank_j} is
\[
\lambda^{-2\left(n-i_{s-1}\right)}DQ_{i_{s}-i_{s-1}}^{\otimes2}DQ_{n-i_{s}}^{\otimes2}\left(v^{1/2}\otimes v^{1/2}\right)(x,x'),
\]
and Lemma \ref{lem:DQDQ_bound} then shows immediately that (\ref{eq:rank_j}) does
indeed hold at rank $s-1$. We point out that the constraint $F\in\mathcal{L}_{v^{1/2},2}$ in the statement of the proposition is imposed because in the case $i_s=n$ we immediately encounter $DQ^{\otimes2}_{n-i_s}(v^{1/2} \otimes v^{1/2})=D(v^{1/2} \otimes v^{1/2}) = v$, and we can control integrals involving $v$ using the drift conditions, as in Lemma \ref{lem:DQDQ_bound}. If we were to give a separate treatment of $\Gamma_{n,x}^{\left(i_{1},...,i_{s}\right)}\left(F\right)$ for coalescent time configurations in which $i_s\neq n$, the constraint on $F$ could be relaxed to a larger function class. 
 
Proceeding with the induction, 
when the hypothesis \eqref{eq:rank_j} holds at rank $j$, we have at rank $j-1$: 
\begin{eqnarray*}
 &  & \lambda^{-2\left(n-i_{j-1}\right)}DQ_{i_{j}-i_{j-1}}^{\otimes2}\ldots DQ_{i_{s}-i_{s-1}}^{\otimes2}DQ_{n-i_{s}}^{\otimes2}\left(v^{1/2}\otimes v^{1/2}\right)(x,x')\\
 &  & \leq c^{s-j+2}v(x)h_{0}(x)\left(\int\check{P}_{i_{j}-i_{j-1}}\left(x,dy\right)v(y)\check{\mathbb{E}}_{y}\left[\prod_{k\in\left\{ i_{j+1}-i_{j},\ldots,i_{s-1}-i_{j}\right\} }v\left(\check{X}_{k}\right)v^{*}\left(\check{X}_{i_{s}-i_{j}}\right)\right]\right)\\
 &  & =c^{s-j+2}v(x)h_{0}(x)\check{\mathbb{E}}_{x}\left[\prod_{k\in\left\{ i_{j}-i_{j-1},\ldots,i_{s-1}-i_{j-1}\right\} }v\left(\check{X}_{k}\right)v^{*}\left(\check{X}_{i_{s}-i_{j-1}}\right)\right],
\end{eqnarray*}
where the inequality follows from applying the induction hypothesis, then multiplying by $v(x^\prime)\geq1$
and then applying Lemma \ref{lem:lam_D_Q_bound} with $\varphi(x)$ the $x$-dependent part of the right hand side of \eqref{eq:rank_j}.  This concludes the inductive proof of
(\ref{eq:claim_induction}).

Consider the case $s>1, i_{1}=0$. Multiplying the right hand side of  (\ref{eq:claim_induction}) by $v(\check{X}_0)=v(x)\geq1$  and recalling the definition of $\Gamma_{n,x}^{(i_1,...,i_s)}$ and $\gamma^N_{0,x}=\delta_x$, we immediately obtain \eqref{eq:numerator_bound}, as desired.
In the case $i_{1}>0$, we multiply  (\ref{eq:claim_induction}) by $v(x')$ and apply Lemma \ref{lem:lam_D_Q_bound} in a similar fashion as before to yield 
\begin{eqnarray*}
 &  & \lambda^{-2n}DQ_{i_{1}}^{\otimes2}DQ_{i_{2}-i_{1}}^{\otimes2}\ldots DQ_{i_{s}-i_{s-1}}^{\otimes2}DQ_{n-i_{s}}^{\otimes2}\left(v^{1/2}\otimes v^{1/2}\right)(x,x')\\
 &  & \leq c^{s+2}v(x)h_{0}(x)\check{\mathbb{E}}_{x}\left[\prod_{k\in\left\{ i_{1},\ldots,i_{s-1}\right\} }v\left(\check{X}_{k}\right)v^{*}\left(\check{X}_{i_{s}}\right)\right]
\end{eqnarray*}
so again we obtain \eqref{eq:numerator_bound} as desired. This completes the treatment of the case $s>1$.

For the case $s=1,i_{1}>0$, 
\begin{eqnarray*}
\lambda^{-2n}Q_{i_{1}}^{\otimes2}DQ_{n-i_{1}}^{\otimes2}\left(v^{1/2}\otimes v^{1/2}\right)(x,x) & = & \lambda^{-2n}DQ_{i_{1}}^{\otimes2}DQ_{n-i_{1}}^{\otimes2}\left(v^{1/2}\otimes v^{1/2}\right)(x,x')\\
 & \leq & cv(x)h_{0}(x)\check{\mathbb{E}}_{x}\left[v^{*}\left(\check{X}_{i_{1}}\right)\right],
\end{eqnarray*}
where the inequality is due to an application of Lemma \ref{lem:DQDQ_bound}. 
Thus we have \eqref{eq:numerator_bound} in the case $s=1,i_{1}>0$.  It only remains to address the case $s=0$, because for the case $s=1,i_{1}=0$ we observe that $\Gamma_{n,x}^{\left(\emptyset\right)}\left(F\right)=\Gamma_{n,x}^{\left(0\right)}\left(F\right)$. 

For $s=0$ we have $\Gamma_{n,x}^{\left(\emptyset\right)}\left(F\right)=\gamma_{n,x}^{\otimes2}\left(F\right)=Q_{n}^{\otimes2}\left(F\right)(x,x)\leq\left\Vert F\right\Vert _{v^{1/2},2}Q_{n}^{\otimes2}(v\otimes v)(x,x)$
and therefore (recall $\check{v}$ from lemma \ref{lem:P_check_M-drift})
\begin{eqnarray}
\lambda^{-2n}h_{0}^{-2}(x)\Gamma_{n,x}^{\left(\emptyset\right)}\left(F\right) & \leq & \left\Vert F\right\Vert _{v^{1/2},2}\lambda^{-2n}h_{0}^{-2}(x)Q_{n}^{\otimes2}(v\otimes v)(x,x)\nonumber \\
 & \leq & c\left\Vert F\right\Vert _{v^{1/2},2}\check{P}_{n}^{\otimes2}\left(\check{v}\otimes\check{v}\right)(x,x)\nonumber \\
 & \leq & c\left\Vert F\right\Vert _{v^{1/2},2}\frac{v(x)}{h_{0}(x)}v^{*}(x).\label{eq:Gamma_empty}
\end{eqnarray}
where the final inequality follows by iteration of the geometric
drift condition \eqref{eq:twisted_geo_drift} and the definition of
$v^{*}$.  Thus \eqref{eq:numerator_bound} holds in the case $s=0$. This completes the proof of the proposition. \end{proof}



\bibliography{SMC_stability}
 \bibliographystyle{plainnat}
\end{document}